\documentclass{article}
\usepackage{pgfplots}
\usepackage{amsmath}
\usepackage{amsthm}
\usepackage{amsfonts}
\usepackage{enumitem}
\usepackage{mathtools}
\usepackage{thmtools}
\usepackage{thm-restate}
\usepackage{framed}
\usepackage{fullpage}

\pgfplotsset{width=7cm}

\newcommand{\reals}{\ensuremath{\mathbb{R}}}
\newcommand{\posreals}{\ensuremath{\reals_{\ge 0}}}

\DeclareMathOperator*{\probOp}{\mathrm{Pr}}
\DeclareMathOperator*{\expectOp}{\mathbb{E}}

\DeclareMathOperator{\restrict}{\mid}

\newcommand{\expect}[2][]{\expectOp_{#1}\left[#2\right]}
\newcommand{\prob}[2][]{\probOp_{#1}\left[#2\right]}

 \newtheorem{theorem}{Theorem}[section]
 
 \newtheorem{lemma}[theorem]{Lemma}
 \newtheorem{corollary}[theorem]{Corollary}
 \theoremstyle{definition}

\author{Maxim Sviridenko\thanks{Yahoo!\ Labs, New York, NY, USA.  \texttt{sviri@yahoo-inc.com}} \and Jan Vondr\'ak\thanks{IBM Almaden Research Center, San Jose, CA, USA.  \texttt{jvondrak@us.ibm.com}} \and Justin Ward\thanks{Department of Computer Science, University of Warwick, Coventry, United Kingdom. \texttt{J.D.Ward@warwick.ac.uk}\enspace Work supported by EPSRC grant EP/J021814/1.}}
\title{Optimal approximation for submodular and supermodular optimization with bounded curvature}

\begin{document}
\maketitle

\pagenumbering{arabic}
\setcounter{page}{1}

\begin{abstract}
We design new approximation algorithms for the problems of optimizing submodular and supermodular functions subject to a single matroid constraint.  Specifically, we consider the case in which we wish to maximize a nondecreasing submodular function or minimize a nonincreasing supermodular function in the setting of bounded total curvature $c$. In the case of submodular maximization with curvature $c$, we obtain a $(1-c/e)$-approximation --- the first improvement over the greedy $(1-e^{-c})/c$-approximation of Conforti and Cornuejols from 1984, which holds for a cardinality constraint, as well as recent approaches that hold for an arbitrary matroid constraint.  

Our approach is based on modifications of the continuous greedy algorithm and non-oblivious local search,  and allows us to approximately maximize the sum of a nonnegative, nondecreasing submodular function and a (possibly negative) linear function.  We show how to reduce both submodular maximization and supermodular minimization to this general problem when the objective function has bounded total curvature.
We prove that the approximation results we obtain are the best possible in the value oracle model, even in the case of a cardinality constraint. 

We define an extension of the notion of curvature to general monotone set functions and show $(1-c)$-approximation for maximization and $1/(1-c)$-approximation for minimization cases.   Finally, we give two concrete applications of our results in the settings of maximum entropy sampling, and the column-subset selection problem.
\end{abstract}

\newcommand{\submod}{g}
\newcommand{\lin}{\ell}
\newcommand{\obj}{f}
\newcommand{\submodExt}{G}
\newcommand{\linExt}{L}
\newcommand{\objExt}{F}
\newcommand{\ground}{X}
\newcommand{\I}{\mathcal{I}}
\newcommand{\M}{\mathcal{M}}
\newcommand{\bases}[1]{\mathcal{B}(#1)}
\newcommand{\conv}{\mathrm{conv}}
\newcommand{\rank}[1]{r_{#1}}
\newcommand{\chr}[1]{\mathbf{1}_{#1}}
\newcommand{\cR}{R}

\newcommand{\linmax}{\hat{v}_\lin}
\newcommand{\submodmax}{\hat{v}_\submod}
\newcommand{\vmax}{\hat{v}}
\newcommand{\emax}{\hat{e}}
\newcommand{\elinmax}{\hat{e}_\lin}
\newcommand{\esubmodmax}{\hat{e}_\submod}

\newcommand{\NOsubmod}{h}
\newcommand{\NOsubmodEst}{\tilde{\NOsubmod}}
\newcommand{\pot}{\psi}
\newcommand{\potEst}{\tilde{\pot}}
\newcommand{\sinit}{S_{0}}

\newcommand{\comp}[1]{\overline{#1}}
\newcommand{\compf}{\comp{f}}
\newcommand{\Mdual}{\M^{*}}
\newcommand{\Bdual}{\B^{*}}
\newcommand{\Odual}{O^{*}}
\newcommand{\Idual}{\I^{*}}
\newcommand{\Sdual}{S^{*}}
\newcommand{\fdual}{f^{*}}
\newcommand{\cinv}{(1 - c)^{-1}}
\newcommand{\NOLS}{\textsc{NOLS}}
\newcommand{\SubmodMax}{\textsc{SubmodMax}}
\newcommand{\SupmodMin}{\textsc{SupmodMin}}

\newcommand{\be}{\mathbf e}
\newcommand{\bc}{\mathbf c}
\newcommand{\bx}{\mathbf x}
\newcommand{\by}{\mathbf y}
\newcommand{\bv}{\mathbf v}
\newcommand{\proj}{\mathsf{proj}}
\newcommand{\dist}{\mathsf{dist}}
\newcommand{\spn}{\mathsf{span}}

\section{Introduction}
\label{sec:introduction}

The problem of maximizing a submodular function subject to various constraints is a meta-problem that appears in various settings, from combinatorial auctions \cite{LLN06,DNS05,Vondrak2008} and viral marketing in social networks \cite{KKT03} to optimal sensor placement in machine learning \cite{KGGK06,KSG08,KRGG09,KG11}. A classic result by Nemhauser, Wolsey and Fisher \cite{Nemhauser1978a} is that the greedy algorithm provides a $(1-1/e)$-approximation for maximizing a nondecreasing submodular function subject to a cardinality constraint. The factor of $1-1/e$ cannot be improved, under the assumption that the algorithm queries the objective function a polynomial number of times \cite{Nemhauser1978}. 

The greedy algorithm has been applied in numerous settings in practice. Although it is useful to know that it never performs worse than $1-1/e$ compared to the optimum, in practice its performance is often even better than this, in fact very close to the optimum. To get a quantitative handle on this phenomenon, various assumptions can be made about the input. One such assumption is the notion of {\em curvature}, introduced by Conforti and Cornu\'ejols \cite{Conforti1984}: A function $f:2^X \rightarrow \reals_+$ has curvature $c \in [0,1]$, if $f(S+j) - f(S)$ does not change by a factor larger than $1-c$ when varying $S$. A function with $c=0$ is linear, so the parameter measures in some sense how far $f$ is from linear. It was shown in \cite{Conforti1984} that the greedy algorithm for nondecreasing submodular functions provides a $(1-e^{-c})/c$-approximation, which tends to $1$ as $c \rightarrow 0$.

Recently, various applications have motivated the study of submodular optimization under various more general constraints. In particular, the $(1-1/e)$-approximation under a cardinality constraint has been generalized to any matroid constraint in \cite{Calinescu2011}. This captures various applications such as welfare maximization in combinatorial auctions \cite{Vondrak2008}, generalized assignment problems \cite{Calinescu2007} and variants of sensor placement \cite{KRGG09}. Assuming curvature $c$, \cite{Vondrak2010} generalized the $(1-e^{-c})/c$-approximation of \cite{Conforti1984} to any matroid constraint, and hypothesized that this is the optimal approximation factor. It was proved in \cite{Vondrak2010} that this factor is indeed optimal for instances of curvature $c$ {\em with respect to the optimum} (a technical variation of the definition, which depends on how values change when measured on top of the optimal solution). 
In the following, we use {\em total curvature} to refer to the original definition of \cite{Conforti1984}, to distinguish from curvature w.r.t.~the optimum \cite{Vondrak2010}.

\subsection{Our Contribution}
\label{sec:our-contribution}

Our main result is that given total curvature $c \in [0,1]$, the $\frac{1-e^{-c}}{c}$-approximation of Conforti and Cornu\'ejols for monotone submodular maximization subject to a cardinality constraint \cite{Conforti1984} is suboptimal and can be improved to a $(1 - c/e - O(\epsilon))$-approximation. We prove that this guarantee holds for the maximization of a nondecreasing submodular function subject to any matroid constraint, thus improving the result of \cite{Vondrak2010} as well. We give two techniques that achieve this result: a modification of the continuous greedy algorithm of \cite{Calinescu2011}, and a variant of the local search algorithm of \cite{Filmus2014}. 

Using the same techniques, we obtain an approximation factor of $1 + \frac{c}{1-c}e^{-1} + \frac{1}{1-c}O(\epsilon)$ for minimizing a nonincreasing supermodular function subject to a matroid constraint. Our approximation guarantees are strictly better than existing algorithms for every value of $c$ except $c = 0$ and $c = 1$.  The relevant ratios are plotted in Figures \ref{fig:submod} and \ref{fig:supmod}.
In the case of minimization, we have also plotted the inverse approximation ratio to aid in comparison.
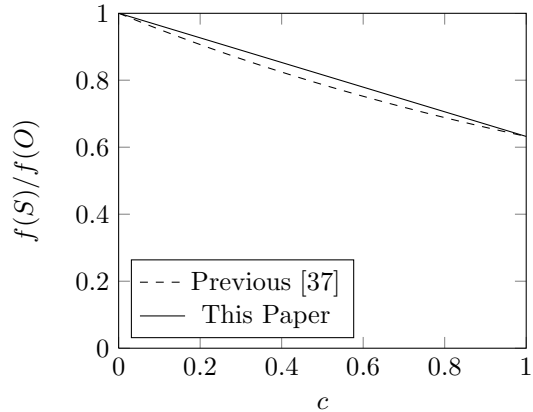
\begin{figure}
\begin{center}
\begin{tikzpicture}
  \begin{axis}[xmin=0,ymin=0,domain=0.001:1, samples=100, smooth, no markers,
    enlargelimits=false, xlabel=$c$, ylabel=${f(S)/f(O)}$, legend pos=south west]
 \addplot gnuplot[black, dashed, id=jan]{(1 - exp(-x))/x};
\addlegendentry{Previous \cite{Vondrak2010}}
 \addplot gnuplot[black, id=us1]{1 - x*exp(-1)};
\addlegendentry{This Paper}
 \end{axis}
 \end{tikzpicture}
\end{center}
\caption{Comparison of Approximation Ratios for Submodular Maximization}
\label{fig:submod}
\end{figure}
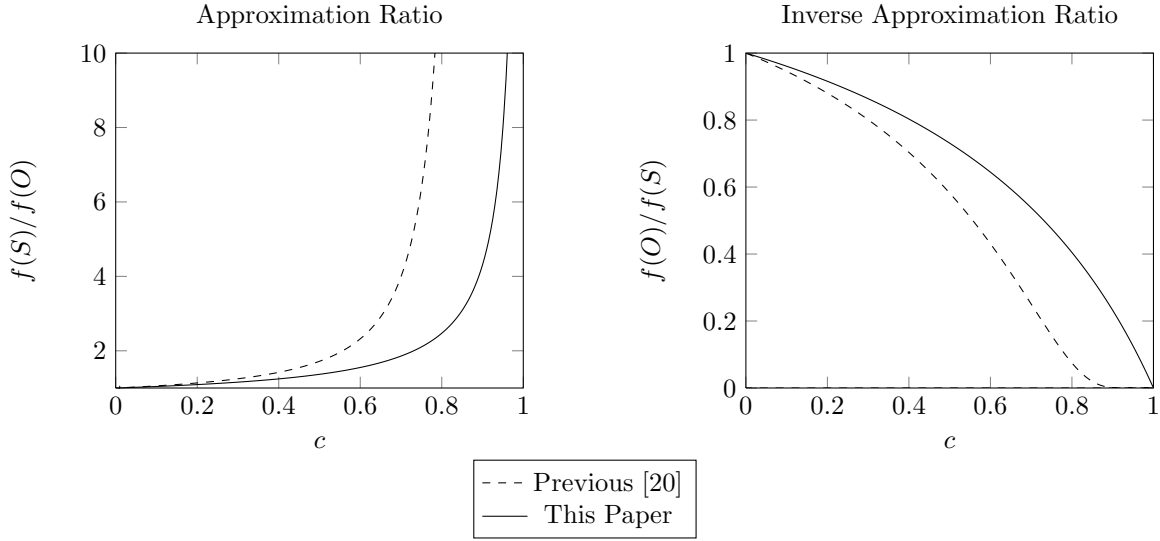
\begin{figure*}
\begin{center}
\begin{tikzpicture}
  \begin{axis}[xmin=0.0,xmax=1.0,ymin=1.0,ymax=10, samples=100, smooth, no markers,
    enlargelimits=false, xlabel=$c$, ylabel=${f(S)/f(O)}$,
    legend columns=1,
    legend entries={Previous \cite{Ilev2001}, This Paper},
    legend to name=named,
    title={Approximation Ratio}]
]
 \addplot gnuplot[black,dashed, domain=0.00:0.9,id=ilev]{(exp(x/(1-x)) - 1)/(x/(1-x))};
 \addplot gnuplot[black,domain=0.00:0.99,id=us2]{1 + exp(-1)*x/(1-x)};
 \end{axis}
\end{tikzpicture}
\hspace{1cm}
\begin{tikzpicture}
  \begin{axis}[xmin=0.0,xmax=1.0,ymin=0.0,ymax=1, samples=1000, smooth, no markers,
    enlargelimits=false, xlabel=$c$, ylabel=${f(O)/f(S)}$, title={Inverse Approximation Ratio}]
 \addplot gnuplot[black,dashed, domain=0.001:1.0,id=ilevinv]{1/((exp(x/(1-x)) - 1)/(x/(1-x)))};
 \addplot gnuplot[black,domain=0.00:0.999,id=us2inv]{1/(1 + exp(-1)*x/(1-x))};
 \end{axis}
 \end{tikzpicture}
\ref{named}
\caption{Comparison of Approximation Ratios for Supermodular Minimization}
\label{fig:supmod}
\end{center}
\end{figure*}


We also derive complementary negative results, showing that no algorithm that evaluates $f$ on only a polynomial number of sets can have approximation performance better than the algorithms we give.  Thus, we resolve the question of optimal approximation as a function of curvature in both the submodular and supermodular case.



Further, we show that the assumption of bounded curvature alone is sufficient to achieve certain approximations, even without assuming submodularity or supermodularity. Specifically, there is a (simple) algorithm that achieves a $(1-c)$-approximation for the maximization of any nondecreasing function of total curvature at most $c$, subject to a matroid constraint. (In contrast, we achieve a $(1-c/e-O(\epsilon))$-approximation with the additional assumption of submodularity.) Also, there is a $\frac{1}{1-c}$-approximation for the minimization of any  nonincreasing function of total curvature at most $c$ subject to a matroid constraint, compared with a $(1 + \frac{c}{1-c}  e^{-1} + \frac{1}{1-c} O(\epsilon))$-approximation for supermodular functions.

\subsection{Applications}
We provide two applications of our results.  In the first application, we are given a positive semidefinite matrix $M$. Let $M[S,S]$ be a principal minor defined by the columns and rows indexed by the set $S\subseteq \{1,\dots,n\}$. In the maximum entropy sampling problem (or more precisely in the generalization of that problem) we would like to find a set $|S|=k$ maximizing $f(S)=\ln \det M[S,S]$. It is well-known that this set function $f(S)$ is submodular \cite{Kelmans1983} (many earlier and alternative proofs of that fact are known). In addition, we know that solving this problem exactly is NP-hard \cite{Lee1996} (see also Lee \cite{Lee2002} for a survey on known optimization techniques for the problem).

We consider the maximum entropy sampling problem when the matrix $M$ has eigenvalues $1\le \lambda_1\le \dots\le \lambda_n$. Since the determinant of any matrix is just a product of its eigenvalues, the Cauchy Interlacing Theorem implies that the  the submodular function $f(S)=\ln \det M[S,S]$ is nondecreasing. In addition we can easily derive a bound on its curvature $c\le 1-1/\lambda_n$ (see the formal definition of curvature in Section \ref{sec:related-work}).  This immediately implies that our new algorithms for submodular maximization have an approximation guarantee of $1 - \left(1-\frac{1}{\lambda_n}\right) \frac{1}{e}$ for the maximum entropy sampling problem.

Our second application is the Column-Subset Selection Problem  arising in various machine learning settings. The goal is, given a matrix $A \in \reals^{m \times n}$, to select a subset of $k$ columns such that the matrix is well-approximated (say in squared Frobenius norm) by a matrix whose columns are in the span of the selected $k$ columns. This is a variant of {\em feature selection}, since the rows might correspond to examples and the columns to features. The problem is to select a subset of $k$ features such that the remaining features can be approximated by linear combinations of the selected features. This is related but not identical to Principal Component Analysis (PCA) where we want to select a subspace of rank $k$ (not necessarily generated by a subset of columns) such that the matrix is well approximated by its projection to this subspace. While PCA can be solved optimally by spectral methods, the Column-Subset Selection Problem is less well understood. Here we take the point of view of approximation algorithms: given a matrix $A$, we want to find a subset of $k$ columns such that the squared Frobenius distance of $A$ from its projection on the span of these $k$ columns is minimized.  To the best of our knowledge, this problem is not known to be NP-hard; on the other hand, the approximation factors of known algorithms are quite large.   The best known algorithm for the problem as stated is a $(k+1)$-approximation algorithm given by Deshpande and Rademacher \cite{Deshpande2010}.  For the related problem in which we may select any set of $r \ge k$ columns that form a rank $k$ submatrix of $A$, Deshpande and Vempala \cite{Deshpande2006} showed that there exist matrices for which $\Omega(k/\epsilon)$ columns must be chosen to obtain a $(1 + \epsilon)$-approximation.  Boutsidis et al.\ \cite{Boutsidis2014} give a matching algorithm, which obtains a set of $O(k/\epsilon)$ columns that give a $(1 + \epsilon)$ approximation.  We refer the reader to \cite{Boutsidis2014} for further background on the history of this and related problems.  

Here, we return to the setting in which only $k$ columns of $A$ may be chosen and show that this is a special case of nonincreasing function minimization with bounded curvature.  We show a relationship between curvature and the condition number $\kappa$ of $A$, which allows us to obtain approximation factor of $\kappa^2$. We define the problem and the related notions more precisely in Section \ref{application}.

\subsection{Related Work}
\label{sec:related-work}

The problem of maximizing a nondecreasing submodular function subject to a cardinality constraint (i.e., a uniform matroid) was studied by Nemhauser, Wolsey, and Fisher \cite{Nemhauser1978a}, who showed that the standard greedy algorithm gives a $(1 - e^{-1})$-approximation.  However, in \cite{Fisher1978}, they show that the greedy algorithm is only $1/2$-approximation for maximizing a nondecreasing submodular function subject to an arbitrary matroid constraint.  More recently, Calinescu et al. \cite{Calinescu2011} obtained a $(1 - e^{-1})$ approximation for an arbitrary matroid constraint.  In their approach, the \emph{continuous greedy algorithm} first maximizes approximately a multilinear extension of the given submodular function and then applies a \emph{pipage rounding} technique inspired by \cite{Ageev2004} to obtain an integral solution.  The running time of this algorithm is dominated by the pipage rounding phase.  Chekuri, Vondr\'{a}k, and Zenklusen \cite{Chekuri2010} later showed that pipage rounding can be replaced by an alternative rounding procedure called \emph{swap rounding} based on the exchange properties of the underlying constraint.  In later work \cite{Chekuri2011,Chekuri2011a}, they developed the notion of a \emph{contention resolution scheme}, which gives a unified treatment for a variety of constraints, and allows rounding approaches for the continuous greedy algorithm to be composed in order to solve submodular maximization problems under combinations of constraints.  Later, Filmus and Ward \cite{Filmus2012} obtained a $(1 - e^{-1})$-approximation for submodular maximization in an arbitrary matroid by using a non-oblivious local search algorithm that does not require rounding.

On the negative side, Nemhauser and Wolsey \cite{Nemhauser1978} showed that it is impossible to improve upon the bound of $(1 - e^{-1})$ in the value oracle model, even under a single cardinality constraint.  In this model, $f$ is given as a value oracle and an algorithm can evaluate $f$ on only a polynomial number of sets.  Feige \cite{Feige1998} showeds that $(1 - e^{-1})$ is the best possible approximation even when the function is given explicitly, unless $P = NP$.  In later work, Vondr\'{a}k \cite{Vondrak2009} introduced the notion of the \emph{symmetry gap} of a submodular function $f$, which unifies many inapproximability results in the value oracle model, and proved new inapproximability results for some specific constrained settings.  Later, Dobzinski and Vondr\'ak \cite{Dobzinski2012} showed how these inapproximability bounds may be converted to matching complexity-theoretic bounds, which hold when $f$ is given explicitly, under the assumption that $RP \neq NP$.

Conforti and Cornu\'{e}jols \cite{Conforti1984} defined the \emph{total curvature} $c$ of non-decreasing submodular function $f$ as
\begin{equation}
\label{eq:curvature}
c = \max_{j \in \ground}\frac{f_{\emptyset}(j) - f_{\ground - j}(j)}{f_{\emptyset(j)}}
= 1 - \min_{j \in \ground} \frac{f_{\ground-j}(j)}{f_{\emptyset}(j)}.
\end{equation}
They showed that greedy algorithm has an approximation ratio of $\frac{1}{1+c}$ for the problem of maximizing a nondecreasing submodular function with curvature at most $c$ subject to a single matroid constraint.  In the special case of a uniform matroid, they were able to show that the greedy is a $\frac{1 - e^{-c}}{c}$-approximation algorithm.  Later, Vondr\'{a}k \cite{Vondrak2010} considered the continuous greedy algorithm in the setting of bounded curvature.  He introduced the notion of \emph{curvature with respect to the optimum}, which is a weaker notion than total curvature, and showed that the continuous greedy algorithm is a $\frac{1 - e^{-c}}{c}$-approximation for maximizing a nondecreasing submodular function $f$ subject to an arbitrary matroid constraint whenever $f$ has curvature at most $c$ with respect to the optimum.  He also showed that it is impossible to obtain a $\frac{1 - e^{-c}}{c}$-approximation in this setting when evaluating $f$ on only a polynomial number of sets.  Unfortunately, unlike total curvature, it is in general not possible to compute the curvature of a function with respect to the optimum, as it requires knowledge of an optimal solution.

We shall also consider the problem of minimizing nonincreasing \emph{supermodular} functions $f : 2^{\ground} \to \posreals$.  By analogy with total curvature, Il'ev \cite{Ilev2001} defines the \emph{steepness} $s$ of a nonincreasing supermodular function.  His definition, which is stated in terms of the marginal \emph{decreases} of the function, is equivalent to \eqref{eq:curvature} when reformulated in terms of marginal gains.  He showed that, in contrast to submodular maximization, the simple  greedy heuristic does not give a constant factor approximation algorithm in the general case.  However, when the supermodular function $f$ has total curvature at most $c$, he shows that the reverse greedy algorithm is an $\frac{e^{p}-1}{p}$-approximation where $p = \frac{c}{1-c}$.

\section{Preliminaries}
\label{sec:preliminaries}

We now fix some of our notation and give two lemmas pertaining to functions with bounded curvature.

\subsection{Set Functions}
A set function $f : 2^\ground \to \posreals$ is \emph{submodular} if $f(A) + f(B) \ge f(A \cap B) + f(A \cup B)$ for all $A, B\subseteq \ground$.  Submodularity can equivalently be characterized in terms of \emph{marginal values}, defined by $f_{A}(i) = f(A + i) - f(A)$ for $i \in \ground$ and $A \subseteq \ground - i$.  Then, $f$ is submodular if and only if $f_A(i) \ge f_B(i)$ for all $A \subseteq B \subseteq \ground$ and $i \not\in B$.    That is, submodular functions are characterized by decreasing marginal values.  

Intuitively, $f$ is \emph{supermodular} if and only if $-f$ is submodular.  That is, $f$ is supermodular if and only if $f_{A}(i) \le f_{B}(i)$ for all $A \subseteq B \subseteq \ground$ and $i \not \in B$.

Finally, we say that a function is \emph{non-decreasing}, or \emph{monotone increasing}, if $f_A(i) \ge 0$ for all $i \in \ground$ and $A \subseteq \ground$, and \emph{non-increasing}, or \emph{monotone decreasing}, if $f_A(i) \le 0$ for all $i \in \ground$ and $A \subseteq \ground$. 

\subsection{Matroids}
\label{sec:matroids}
We now present the definitions and notations that we shall require when dealing with matroids.  We refer the reader \cite{Schrijver2003} for a detailed introduction to basic matroid theory.  Let $\M=(\ground,\I)$ be a matroid defined on ground set $\ground$ with independent sets given by $\I$.  We denote by $\bases{\M}$ the set of all bases (inclusion-wise maximal sets in $\I$) of $\M$.  We denote by $P(\M)$ the matroid polytope for $\M$, given by:
\begin{equation*}
P(\M) = \conv \{ \chr{I}\,:\,I \in \I \} 
= \{ x \ge 0\,:\,\sum_{j \in S} x_{j} \le \rank{\M}(S),\ \forall S \subseteq \ground \},
\end{equation*}
where $\rank{\M}$ denotes the rank function associated with $\M$.
The second equality above is due to Edmonds \cite{Edmonds1971}.  Similarly, we denote by $B(\M)$ the base polytope associated with $\M$:
\begin{equation*}
B(\M) = \conv \{ \chr{I}\,:\,I \in \bases{\M} \} = \{ x \in P(\M)\,: \sum_{j \in \ground}x_{j} = \rank{\M}(\ground) \}.
\end{equation*}

For a matroid $\M = (\ground,\I)$, we denote by $\Mdual$ the dual system $(\ground,\Idual)$ whose independent sets $\Idual$ are defined as those subsets $A \subseteq \ground$ that satisfy $A \cap B = \emptyset$ for some $B \in \I$ (i.e., those subsets that are disjoint from some base of $\M$.  Then, a standard result of matroid theory shows that $\Mdual$ is a matroid whenever $\M$ is a matroid, and, moreover,  $\bases{\Mdual}$ is precisely the set $\{ \ground \setminus B : B \in \bases{\M} \}$ of complements of bases of $\M$.

Finally, given a set of elements $D \subseteq \ground$, we denote by $\M \restrict D$ the matroid $(\ground \cap D, \I')$ obtained by restricting to $\M$ to $D$.  The independent sets $\I'$ of $\M \restrict D$ are simply those independent sets of $\M$ that contain only elements from $D$.  That is, $I' = \{A \in \I :  A \cap D = A\}$.

\subsection{Lemmas for Functions with Bounded Curvature}
\label{sec:curvature}

We now give two general lemmas pertaining to functions of bounded curvature that will be useful in our analysis.  The proofs, which follow directly from \eqref{eq:curvature}, are given in the Appendix.
\begin{restatable}{lemma}{submodcurv}
\label{lem:submod-curv}
If $f : 2^{\ground} \to \posreals$ is a monotone increasing submodular function with total curvature at most $c$, then 
$\sum_{j \in A}f_{\ground - j}(j) \ge (1-c)f(A)$ for all $A \subseteq \ground$.
\end{restatable}
\begin{restatable}{lemma}{supmodcurv}
\label{lem:supmod-curv}
If $f : 2^{\ground}\to \posreals$ is a monotone decreasing supermodular function with total curvature at most $c$, then 
$(1-c)\sum_{j \in A}f_{\emptyset}(j) \ge -f(\ground \setminus A)$ for all $A \subseteq \ground$.
\end{restatable}

\section{Submodular + Linear Maximization}
\label{sec:overv-our-appr}

Our new results for both submodular maximization and supermodular minimization with bounded curvature make use of an algorithm for the following meta-problem: we are given a monotone increasing, nonnegative, submodular function $\submod : 2^{\ground} \to \posreals$, a linear function $\lin:2^{\ground}\to \reals$, and a matroid  $\M=(\ground,\I)$ and must find a base $S \in \bases{\M}$ maximizing $\submod(S) + \lin(S)$.  Note that we do not require $\lin$ to be nonnegative.  Indeed, in the case of supermodular minimization (discussed in Section \ref{sec:superm-minim}), our approach shall require that $\lin$ be a negative, monotone decreasing function.  For $j \in \ground$, we shall write $\lin(j)$ and $\submod(j)$ as a shorthand for $\lin(\{j\})$ and $\submod(\{j\})$.  We note that because $\lin$ is linear, we have $\lin(A) = \sum_{j \in A}\lin(j)$ for all $A \subseteq \ground$.

Let $\submodmax = \max_{j \in \ground}\submod(j)$, $\linmax = \max_{j \in \ground}|\lin(j)|$, and $\vmax = \max(\submodmax,\linmax)$.  Then, because $\submod$ is submodular and $\lin$ is linear, we have both $\submod(A) \le n\vmax$ and $|\lin(A)| \le n\vmax$ for every set $A \subseteq \ground$.  Moreover, given $\lin$ and $\submod$, we can easily compute $\vmax$ in time $O(n)$.  Our main technical result is the following, which gives a joint approximation for $\submod$ and $\lin$.

\begin{theorem}
\label{thm:main}
For every $\epsilon > 0$, there is an algorithm that, given a monotone increasing submodular function $g:2^X \rightarrow \posreals$, a linear function $\ell:2^X \rightarrow \reals$ and a matroid $\M$, produces a set $S \in \bases{M}$ in polynomial time satisfying
\begin{equation*}
\submod(S) + \lin(S) \ge \left(1-e^{-1}\right)\submod(O) + \lin(O) - O(\epsilon)\cdot \vmax,
\end{equation*}
for every $O \in \bases{M}$, with high probability.
\end{theorem}
In the next two sections, we give two different algorithms satisfying the conditions of Theorem \ref{thm:main}.

\section{A Modified Continuous Greedy Algorithm}
\label{sec:cont-greedy}
The first algorithm we consider is a modification of the continuous greedy algorithm of \cite{Calinescu2011}
We first sketch the algorithm conceptually in the continuous setting, ignoring certain technicalities.  

\subsection{Overview of the Algorithm}
Consider $x\in [0,1]^{X}$.  For any function $f : 2^{X}\to \reals$, the \emph{multilinear extension}  of $f$ is a function $F : [0,1]^{X} \to \reals$ given by $F(x) = \expect{f(\cR_{x})}$, where $\cR_{x}$ is a random subset of $\ground$ in which each element $e$ appears independently with probability $x_{e}$.  We let $G$ denote the multilinear extension of the given, monotone increasing submodular function $\submod$, and $L$ denote the multilinear extension of the given linear function $\lin$.
Note that due to the linearity of expectation, $L(x) = \expect{\lin(\cR_{x})} = \sum_{j \in X}x_{j}\lin(j)$.  That is, the multilinear extension $L$ corresponds to the natural, linear extension of $\lin$.  Let $P(\M)$ and $B(\M)$ be the matroid polytope and matroid base polytope associated $\M$, and let $O$ be the arbitrary base in $\bases{\M}$ to which we shall compare our solution in Theorem \ref{thm:main}.
\begin{figure}
\begin{framed}
\noindent {\bf Modified Continuous Greedy}
\begin{itemize}[itemsep=0ex, topsep=0.5ex]
\item Guess the values of $\lambda = \lin(O)$ and $\gamma = \submod(O)$.
\item Initialize $x = {\bf 0}$.
\item For time running from $t=0$ to $t=1$, update $x$ according to $\frac{dx}{dt} = v(t)$,
where $v(t) \in P(\M)$ is a vector satisfying both:
\begin{align*}
L(v(t)) &\geq \lambda \\
v(t) \cdot \nabla G(x(t)) &\geq \gamma - G(x(t)).
\end{align*}
Such a vector exists because $\chr{O}$ is one possible candidate (as in the analysis of \cite{Calinescu2011}).
\item Apply pipage rounding to the point $x(1)$ and return the resulting solution.
\end{itemize}
\end{framed}
\caption{The modified continuous greedy algorithm}
\label{fig:cont-greedy-alg}
\end{figure}
Our algorithm is shown in Figure \ref{fig:cont-greedy-alg}.
In contrast to the standard continuous greedy algorithm, in the third step we require a direction that is larger than \emph{both} the value of $\ell(O)$ and the residual value $\gamma - G(x(t))$.  Applying the standard continuous greedy algorithm to $(\submod + \lin)$ gives a direction that is larger than the \emph{sum} of these two values, but this is insufficient for our purposes.  

Our analysis proceeds separately for $L(x)$ and $G(x)$. First, because $L$ is linear, we obtain
$$ \frac{dL}{dt}(x(t)) = L\left(\frac{dx}{dt}\right) \geq \lambda, $$
for every time $t$, and hence $L(x(1)) \geq \lambda = \lin(O)$. For the submodular component, we obtain
$$ \frac{dG}{dt}(x(t)) = \frac{dx}{dt} \cdot \nabla G(x(t)) \geq \gamma - G(x(t)) $$
similar to the analysis in \cite{Calinescu2011}. This leads to a differential equation that gives
$$G(x(t)) \geq (1 - e^{-t}) \gamma.$$

The final pipage rounding phase is oblivious to the value of the objective function and so is not affected by the potential negativity of $\ell$.  We can view $\submodExt + \linExt$ as the multilinear extension of $\submod + \lin$ and so, as in the standard continuous greedy analysis, pipage rounding produces an integral solution $S$ satisfying $\submod(S) + \lin(S) \ge \submodExt(x(1)) + \linExt(x(1)) \ge (1 - e^{-t})\submod(O) + \lin(O)$.

\subsection{Implementation of the Modified Continuous Greedy Algorithm}

Now we discuss the technical details of how the continuous greedy algorithm can be implemented efficiently.
There are three main issues that we ignored in our previous discussion: (1) How do we ``guess" the values of $\lin(O)$ and $\submod(O)$; 
(2) How do we find a suitable direction $v(t)$ in each step of the algorithm; and 
(3) How do we discretize time efficiently?  Let us now address them one by one.

\paragraph{Guessing the optimal values:}
In fact, it is enough to guess the value of $\lin(O)$; we will optimize over $v(t) \cdot \nabla \submodExt$ later.  
Recall that $|\lin(O)| \le n\vmax$.  We discretize the interval\footnote{In the applications we consider $\lin$ is either nonnegative or non-positive, and so we need only consider half of the given interval.  For simplicity, here we give a general approach that does not depend on the sign of $\lin$.  In general, we have favored, whenever possible, simplicity in the analysis over obtaining the best runtime bounds.} $[-n\vmax, n\vmax]$ with $O(\epsilon^{-1})$ points of the form $i \epsilon \cdot \vmax$ for $-\epsilon^{-1} \le i \le \epsilon^{-1}$, filling the interval $[-\vmax, \vmax]$, together with $O(\epsilon^{-1}n\log n)$ points of the form $(1 + \epsilon/n)^{i}\cdot\vmax$ and $-(1 + \epsilon/n)^{i}\cdot\vmax$ for $0 \le i \le \log_{1 + \epsilon/n} n$, filling the intervals $[\vmax, n\vmax]$ and $[-n\vmax, -\vmax]$, respectively.  We then run the following algorithm using each point as a guess for $\lambda$, and return the best solution found.  Then if $|\lin(O)| < \vmax$, we must have  
$$\lin(O) \ge \lambda \ge \lin(O) - \epsilon\cdot\vmax,$$ for some iteration (using one of the guesses in $[-\vmax, \vmax]$).  Similarly, if $|\lin(O)| \ge \vmax$, then for some iteration we have
$$\lin(O) \ge \lambda \ge \lin(O) - (\epsilon/n)|\lin(O)| \ge \lin(O) - \epsilon\cdot\vmax$$
(using one of the guesses in $[\vmax, n\vmax]$ or $[-n\vmax, -\vmax]$).  For the remainder of our analysis we consider this particular iteration.

\paragraph{Finding a suitable direction:}
Given our guess of $\lambda$ and a current solution $x(t)$, our goal is to find a direction $v(t) \in P(\M)$ such that
$L(v(t)) \geq \lambda$ and $v(t) \cdot \nabla \submodExt(x(t)) \geq \gamma - \submodExt(x(t))$.  As in \cite{Calinescu2011}, we must estimate $\nabla \submodExt(x(t))$ by random sampling.   Then, given an estimate $\tilde{\nabla}G$, we solve the linear program:
\begin{equation}
\max \{ v \cdot \tilde{\nabla} G:  v \in B(\M),  L(v) \geq \lambda \}.\label{eq:lp}
\end{equation}
We can do this by the ellipsoid method, for example (or more efficiently using other methods).

Following the analysis of \cite{Calinescu2011}, we can obtain, in polynomial time, an estimate satisfying
\begin{equation*}
\label{eq:cg-error-bound}
|v(t) \cdot \tilde{\nabla} \submodExt(x(t)) - v(t) \cdot \nabla \submodExt(x(t))| \le O(\epsilon) \cdot \vmax,
\end{equation*}
with high probability.  Since $L(O) = \lin(O) \ge \lambda$, the base $O$ is a feasible solution of \eqref{eq:lp}.  Because $G(x(t))$ is concave along any nonnegative direction, we then have:
\begin{equation}
\label{eq:cg-submod-bound}
v \cdot \nabla G(x(t)) \ge \submod(O) - \submodExt(x(t)) - O(\epsilon)\cdot \vmax,
\end{equation}
just as in the analysis of \cite{Calinescu2011}.

\paragraph{Discretizing the algorithm:}
We discretize time into steps of length $\delta=\epsilon/n^{2}$; let us assume for simplicity that $1/\delta$ is an integer.
In each step, we find a direction $v(t)$ as described above and we update
\[
x(t+\delta) = x(t) + \delta \cdot v(t).
\]
Clearly, after $1/\delta$ steps we obtain a solution $x(1)$ which is a convex combination of points $v(t) \in P(\M)$, and therefore a feasible solution. In each step, we had $\linExt(v(t)) \geq \lambda$ and so
\[
\linExt(x(1)) = \sum_{i=1}^{1/\delta} \delta \cdot \linExt(v(i\cdot \delta)) \geq \lambda \geq \lin(O) - \epsilon\cdot \vmax.
\]
The analysis of the submodular component follows along the lines of \cite{Calinescu2011}.  In one time step, we gain
\begin{align*}
\submodExt(x(t&+\delta)) - \submodExt(x(t)) \\
& \geq (1 - n\delta) \cdot (x(t+\delta) - x(t)) \cdot \nabla \submodExt(x(t)) \\
& \geq (1 - n\delta)\cdot \delta \cdot [\submod(O) - \submodExt(x(t)) - O(\epsilon)\cdot \vmax] \\
& \geq \delta\cdot [\submod(O) - \submodExt(x(t)) - O(\epsilon)\cdot \vmax - (\epsilon/n)\cdot \submod(O)] \\
& = \delta\cdot [\submod(O) - \submodExt(x(t)) - O(\epsilon)\cdot \vmax]
\end{align*}
using the bound \eqref{eq:cg-submod-bound}. By induction (as in \cite{Calinescu2011}), we obtain
$$ \submodExt(x(t)) \geq (1 - e^{-t})\submod(O) - O(\epsilon)\cdot\vmax,$$
and so $$\submodExt(x(1)) + \linExt(x(1)) \ge (1 - e^{-1})\submod(O) + \lin(O) - O(\epsilon)\cdot\vmax,$$ as required.

\section{Non-Oblivious Local Search}
\label{sec:non-oblivious-local}

We now give another proof of Theorem \ref{thm:main}, using a modification of the local search algorithm of \cite{Filmus2014}.  In contrast to the modified continuous greedy algorithm, our modified local search algorithm does not need to guess the optimal value of $\ell(O)$, and also does not need to solve the associated continuous optimization problem given in \eqref{eq:lp}.  However, here the convergence time of the algorithm becomes an issue that must be dealt with.

\subsection{Overview of the Algorithm}
\label{sec:overv-non-obliv}

We begin by presenting a few necessary lemmas and definitions from the analysis of \cite{Filmus2014}.  We shall require the following general property of matroid bases, first proved by Brualdi \cite{Brualdi1969}, which can also be found in e.g. \cite[Corollary 39.12a]{Schrijver2003}.
\begin{lemma}
\label{lem:brualdi}
Let $\M$ be a matroid and $A$ and $B$ be two bases in $\bases{\M}$.  Then, there exists a bijection $\pi : A \to B$ such that $A - x + \pi(x) \in \bases{\M}$ for all $x \in A$.
\end{lemma}
We can restate Lemma \ref{lem:brualdi} as follows: let $A = \{a_{1},\ldots,a_{r}\}$ and $B$ be bases of a matroid $\M$ of $r$.  Then we can index the elements of $b_{i} \in B$ so that $b_{i} = \pi(a_{i})$, and then we have that $A - a_{i} + b_{i} \in \bases{\M}$ for all $1 \le i \le r$.  The resulting collection of sets $\{A - a_{i} +b_{i}\}_{i \in A}$ will define a set of feasible swaps between the bases $A$ and $B$ that we consider when analyzing our local search algorithm.

The local search algorithm of \cite{Filmus2014} maximizes a monotone submodular function $\submod$ using a simple local search routine that evaluates the quality of the current solution using an auxiliary potential $\NOsubmod$, derived from $\submod$ as follows:
\begin{equation*}
\NOsubmod(A) = \sum_{B \subseteq A}\submod(B) \cdot \int_{0}^{1}\frac{e^{p}}{e - 1}\cdot p^{|B| - 1}(1 - p)^{|A| - |B|} dp.
\end{equation*}

We defer a discussion of issues related to convergence and computing $\NOsubmod$ until the next subsection, and first sketch the main idea of our modified algorithm.  We shall make use of the following fact, proved in \cite[Lemma 4.4, p.\ 524-5]{Filmus2014}: for all $A$,
$$\submod(A) \le \NOsubmod(A)\le C \cdot \submod(A) \ln n,$$ for some constant $C$.

In order to jointly maximize $\submod(S) + \lin(S)$, we employ a modified local search algorithm that is guided by the potential $\pot$, given by:
\begin{equation*}
\pot(A) = (1 - e^{-1})\NOsubmod(A) + \lin(A).
\end{equation*}
Our final algorithm is shown in Figure \ref{fig:non-oblv-alg}.

\begin{figure}
\begin{framed}
\noindent {\bf Non-Oblivious Local Search}
\begin{itemize}[itemsep=0ex, topsep=0.5ex]
\item Let $\delta = \frac{\epsilon}{n}\cdot\vmax$.
\item $S \gets $ an arbitrary base $\sinit \in \bases{\M}$.
\item  While there exists $a \in S$ and $b \in \ground \setminus S$
such that $S - a + b \in \bases{\M}$ and $$\pot(S - a + b) \ge \pot(S) + \delta,$$ set $S \gets S - a + b$.
\item Return $S$
\end{itemize}
\end{framed}
\caption{The non-oblivious local search algorithm}
\label{fig:non-oblv-alg}
\end{figure}

The following Lemma shows that if it is impossible to significantly improve $\pot(S)$ by exchanging a single element, then both $\submod(S)$ and $\lin(S)$ must have relatively high values.

\begin{lemma}
\label{lem:locality-gap}
Let $A = \{a_{1},\ldots,a_{r}\}$ and $B = \{b_{1},\ldots,b_{r}\}$ be any two bases of a matroid $\M$, and suppose that the elements of $B$ are indexed according to Lemma \ref{lem:brualdi} so that $A - a_{i} + b_{i} \in \bases{\M}$ for all $1 \le i \le r$.  Then,
\begin{equation*}
\submod(A) + \lin(A) 
\ge (1 - e^{-1})\submod(B) + \lin(B) + \sum_{i=1}^{r}\left[\pot(A) - \pot(A\!-\!a_{i}\!+\!b_{i})\right].
\end{equation*}
\end{lemma}
\begin{proof}
Filmus and Ward \cite[Theorem 5.1, p.\ 526]{Filmus2014} show that for any submodular function $\submod$, the associated function $\NOsubmod$ satisfies
\begin{equation}
\label{eq:filmus-locality-gap}
\frac{e}{e - 1}\submod(A) \ge \submod(B) + \sum_{i=1}^{r}\left[\NOsubmod(A) - \NOsubmod(A - a_{i} + b_{i})\right].
\end{equation}
We note that since $\lin$ is linear, we have:
\begin{align}
\label{eq:lin-locality-gap}
\lin(A)
&= \lin(B) + \sum_{i = 1}^{r}\left[\lin(a_{i}) - \lin(b_{i})\right] \notag \\
&= \lin(B) + \sum_{i = 1}^{r}\left[\lin(A) - \lin(A - a_{i} + b_{i})\right] 
\end{align}
Adding $(1 - e^{-1})$ times \eqref{eq:filmus-locality-gap} to \eqref{eq:lin-locality-gap} then completes the proof.
\end{proof}

Suppose that $S \in \bases{\M}$ is locally optimal for $\pot$ under single-element exchanges, and let $O$ be an arbitrary base of $\M$.  Then, local optimality of $S$ implies that $\pot(S) - \pot(S - s_{i} + o_{i}) \ge 0$ for all $i \in [r]$, where the elements $s_{i}$ of $S$ and $o_{i}$ of $O$ have been indexed according to Lemma \ref{lem:brualdi}.  Then, Lemma \ref{lem:locality-gap} gives $\submod(S) + \lin(S) \ge \left(1 - e^{-1}\right)\submod(O) + \lin(O)$, as required by Theorem \ref{thm:main}.

\subsection{Implementation of the Non-Oblivious Local Search Algorithm}
\label{sec:impl-non-obliv}

We now show how to obtain a polynomial-time algorithm from Lemma \ref{lem:locality-gap}.  We face two technical difficulties:
(1) how do we compute $\pot$ efficiently in polynomial time; and (2) how do we ensure that the search for improvements converges to a local optimum in polynomial time?  As in the case of the continuous greedy algorithm, we can address these issues by using standard techniques, but we must be careful since $\lin$ may take negative values.  As in that case, we have not attempted to obtain the most efficient possible running time analysis here, focusing instead on simplifying the arguments.

\paragraph{Estimating $\pot$ efficiently:}

Although the definition of $\NOsubmod$ requires evaluating $\submod$ on a potentially exponential number of sets, Filmus and Ward show that $\NOsubmod$ can be estimated efficiently using a sampling procedure:
\begin{lemma}[{\cite[Lemma 5.1, p.\ 525]{Filmus2014}}]
\label{lem:NOsubmodEst}
Let $\NOsubmodEst(A)$ be an estimate of $\NOsubmod(A)$ computed from $N\!=\Omega(\epsilon^{-2}\ln^{2}n\ln M)$ samples of $\submod$.  Then, 
\[
\prob{|\NOsubmodEst(A) - \NOsubmod(A)| \ge \epsilon \cdot \NOsubmod(A)} = O(M^{-1}).
\]
\end{lemma}
We let $\potEst(A) = \NOsubmodEst(A) + \lin(A)$ be an estimate of $\pot$.  Set $\delta = \frac{\epsilon}{n}\cdot \vmax$.
We shall ensure that $\potEst(A)$ differs from $\pot(A)$ by at most
\begin{align*}
\delta = \frac{\epsilon}{n}\cdot \vmax \ge \frac{\epsilon}{n^{2}}\cdot \submod(A) 
&= \frac{\epsilon}{C\cdot n^{2} \ln n}\cdot C\cdot \submod(A) \ln n \\
&\ge \frac{\epsilon}{C\cdot n^{2} \ln n}\cdot \NOsubmod(A).
\end{align*}
Applying Lemma \ref{lem:NOsubmodEst}, we can then ensure that
\[
\prob{|\potEst(A) - \pot(A)| \ge \delta} = O(M^{-1}),
\]
by using $\Omega(\epsilon^{-2}n^{4}\ln^{4}n\ln M)$ samples for each computation of $\pot$.  By the union bound, we can ensure that $|\potEst(A) - \pot(A)| \le \delta$ holds with high probability for all sets $A$ considered by the algorithm, by setting $M$ appropriately.  In particular, if we evaluate $\potEst$ on any polynomial number of distinct sets $A$, it suffices to make $M$ polynomially small, which requires only  a polynomial number of samples for each evaluation.

\paragraph{Bounding the convergence time of the algorithm:}

We initialize our search with an arbitrary base $\sinit \in \bases{\M}$, and at each step of the algorithm, we restrict our search to those improvements that yield a significant increase in the value of $\pot$.  Specifically, we require that each improvement increases the current value of $\pot$ by at least an additive term $\delta = \frac{\epsilon}{n}\cdot \vmax$.  We now bound the total number of improvements made by the algorithm.

We suppose that all values $\potEst(A)$ computed by the algorithm satisfy
\[\pot(A) - \delta \le \potEst(A) \le \pot(A) + \delta.\]  From the previous discussion, we can ensure that this is indeed the case with high probability.

Let $O_{\pot} = \arg\max_{A \in \bases{\M}}\pot(A)$.  Then, the total number of improvements applied by the algorithm is at most:
\begin{align*}
&{\textstyle \frac{1}{\delta}}(\potEst(O_{\pot}) - \potEst(\sinit)) \\
&\le {\textstyle \frac{1}{\delta}}(\pot(O_{\pot}) - \pot(\sinit) + 2\delta)
\\ &= {\textstyle \frac{1}{\delta}}((1 - e^{-1})\!\cdot\!(\NOsubmod(O_{\pot}) - \NOsubmod(\sinit)) + \lin(O_{\pot}) - \lin(\sinit)  + 2\delta)
\\ &\le {\textstyle \frac{1}{\delta}}((1 - e^{-1})\!\cdot\!\NOsubmod(O_{\pot}) + |\lin(O_{\pot})| + |\lin(\sinit)| + 2\delta)
\\ &\le {\textstyle \frac{1}{\delta}}((1 - e^{-1})\!\cdot\!C\!\cdot\!\submod(O_{\pot})\ln n + |\lin(O_{\pot})| + |\lin(\sinit)| + 2\delta)
\\ &\le {\textstyle \frac{1}{\delta}}((1 - e^{-1})\!\cdot\!C\!\cdot\!\vmax\!\cdot\!n \ln n + n\!\cdot\!\vmax + n\!\cdot\!\vmax + 2\delta)
\\ &= O(\epsilon^{-1}n^{2}\ln n).
\end{align*}
Each improvement step requires $O(n^{2})$ evaluations of $\pot$.  From the discussion in the previous section, setting $M$ sufficiently high will ensures that all of the estimates made for the first $\Omega(\epsilon^{-1}n^{2}\ln n)$ iterations will satisfy our assumptions with high probability, and so the algorithm will converge in polynomial time.  

In order to obtain a deterministic bound on the running time of the algorithm we simply terminate our search if it has not converged in $\Omega(\epsilon^{-1}n^{2}\ln n)$ steps and return the current solution.  Then when the resulting algorithm terminates, with high probability, we indeed have $\potEst(S) - \potEst(S - s_{i} + o_{i}) \le \delta$ for every $i \in [r]$ and so
 \begin{equation*} 
\sum_{i = 1}^{r}[\pot(S) - \pot(S - s_{i} + o_{i})] \le r(\delta + 2\delta) \le 3\epsilon \cdot \vmax.
\end{equation*}
From Lemma \ref{lem:locality-gap}, the set $S$ produced by the algorithm then satisfies
\[
\submod(S) + \lin(S) \ge (1 - e^{-1})\submod(O) + \lin(O) - O(\epsilon) \cdot \vmax,
\]
as required by Theorem \ref{thm:main}.

\section{Submodular Maximization and Supermodular Minimization}
\label{sec:applications}
We now return to the problems of submodular maximization and supermodular minimization with bounded curvature.  We reduce both problems to the general setting introduced in Section \ref{sec:overv-our-appr}.  In both cases, we suppose that we are seeking optimize a function $\obj : 2^{\ground} \to \posreals$ over a given matroid $\M = (\ground,\I)$ and we let $O$ denote any optimal base of $\M$ (i.e., a base of $\M$ that either maximizes or minimizes $\obj$, according to the setting).

\subsection{Submodular Maximization}
\label{sec:subm-maxim}

Suppose that $\obj$ is a monotone increasing submodular function with curvature at most $c \in [0,1]$, and we seek to maximize $\obj$ over a matroid $\M$. 

\begin{theorem}
\label{thm:submod-main}
For every $\epsilon > 0$ and $c \in [0,1]$, there is an algorithm that given a monotone increasing submodular function $f:2^X \rightarrow \posreals$ of curvature $c$ and a matroid $\M = (\ground,\I)$, produces a set $S \in \I$ in polynomial time satisfying
\begin{equation*}
\obj(S) \geq (1 - c/e - O(\epsilon)) \obj(O)
\end{equation*}
for every $O \in \I$, with high probability.
\end{theorem}

\begin{proof}
Define the functions:
\begin{align*}
\lin(A) &= \sum_{j \in A}\obj_{\ground - j}(j) \\
\submod(A) &= \obj(A) - \lin(A).
\end{align*}
Then, $\lin$ is linear and $\submod$ is submodular, monotone increasing, and nonnegative (as verified in Lemma \ref{lem:submod-app} of the appendix).  Moreover, because $\obj$ has curvature at most $c$, Lemma \ref{lem:submod-curv} implies that for any set $A \subseteq \ground$, $$\textstyle \lin(A) = \sum_{j \in A}\obj_{\ground - j}(j) \ge (1 - c)\obj(A).$$

In order to apply Theorem \ref{thm:main} we must bound the term $\vmax$.  By optimality of $O$ and non-negativity of $\lin$ and $\submod$, we have $\vmax \le \submod(O) + \lin(O) \le \obj(O)$.  From Theorem \ref{thm:main}, we can find a solution $S$ satisfying:
\begin{align*}
\obj(S) &= \submod(S) + \lin(S) \\
&\ge \left(1 - e^{-1}\right)\submod(O) + \lin(O) - O(\epsilon)\cdot \obj(O) \\
           &= \left(1 - e^{-1}\right)\obj(O) + e^{-1}\lin(O) - O(\epsilon) \cdot \obj(O) \\
           &\ge \left(1 - e^{-1}\right)\obj(O) + (1-c)e^{-1}\obj(O) - O(\epsilon) \cdot\obj(O) \\
           &= \left(1 - ce^{-1} - O(\epsilon)\right)\obj(O).
\end{align*}
\end{proof}

\subsection{Supermodular Minimization}
\label{sec:superm-minim}

Suppose that $\obj$ is a monotone decreasing supermodular function with curvature at most $c \in [0,1)$ and we seek to minimize $\obj$ over a matroid $\M$. 

\begin{theorem}
\label{thm:supermod-main}
For every $\epsilon > 0$ and $c \in [0,1)$, there is an algorithm that given a monotone decreasing supermodular function $f:2^X \rightarrow \posreals$ of curvature $c$ and a matroid $\M = (\ground,\I)$, produces a set $S \in \I$ in polynomial time satisfying
\begin{equation*}
\obj(S) \leq \left(1 + \frac{c}{1-c} \cdot e^{-1} + \frac{1}{1-c}\cdot O(\epsilon) \right) \obj(O)
\end{equation*}
for every $O \in \I$, with high probability.
\end{theorem}

\begin{proof}
Define the linear and submodular functions:
\begin{align*}
\lin(A) &= \sum_{j \in A}\obj_{\emptyset}(j) \\
\submod(A) &= -\lin(A) - \obj(\ground \setminus A).
\end{align*}
Because $f$ is monotone decreasing, we have $\obj_{\emptyset}(j) \le 0$ and so $\lin(A) \le 0$ for all $A \subseteq \ground$.  Thus, $\lin$ is a non-positive, decreasing linear function.  However, as we verify in Lemma \ref{lem:supmod-app} of the appendix, $\submod$ is submodular, monotone increasing, and nonnegative.

We shall consider the problem of \emph{maximizing} $\submod(S) + \ell(S) = -\obj(\ground \setminus S)$ in the \emph{dual} matroid $\Mdual$, whose bases correspond to complements of bases of $\M$.  We compare our solution $S$ to this problem to the base $\Odual = \ground \setminus O$ of $\Mdual$.  Again, in order to apply Theorem \ref{thm:main}, we must bound the term $\vmax$.  Here, because $\lin(A)$ is non-positive, we cannot bound $\vmax$ directly as in the previous section.  Rather, we proceed by partial enumeration.  Let $\emax = \arg\max_{j \in \Odual}\max(\submod(j),|\lin(j)|)$.  We iterate through all possible guesses $e \in \ground$ for $\emax$, and for each such $e$ consider $\vmax_{e} = \max(\submod(e),|\lin(e)|)$.  We set $\ground_{e}$ to be the set $\{j \in \ground : \submod(j) \le \vmax_{e}\ \land\ |\lin(j)| \le \vmax_{e} \}$, and consider the matroid $\Mdual_{e} = \Mdual \restrict \ground_{e}$, obtained by restricting $\Mdual$ to the ground set $\ground_{e}$.  For each $e$ satisfying $\rank{\Mdual_{e}}(\ground_{e}) = \rank{\Mdual}(\ground)$, we apply our algorithm to the problem $\max\{\submod(A) + \lin(A) : A \in \Mdual_{e}\}$, and return the best solution $S$ obtained.  Note since $\rank{\Mdual_{e}}(\ground_{e}) = \rank{\Mdual}(\ground)$, the set $S$ is also a base of $\Mdual$ and so $X \setminus S$ is a base of $\M$.  

Consider the iteration in which we correctly guess $e = \emax$.  In the corresponding restricted instance we have $\submod(j) \le \vmax_{e} = \vmax$ and $|\lin(j)| \le \vmax_{e} = \vmax$ for all $j \in \ground_{e}$.  Additionally, $\Odual \subseteq \ground_{e}$ and so $\Odual \in \bases{\M_{e}}$, as required by our analysis.   Finally, from the definition of $\submod$ and $\lin$, we have $\obj(O) = -\lin(\Odual) - \submod(\Odual)$.  Since $\emax \in \Odual$, and $\lin$ is nonpositive while $\obj$ is nonnegative, 
\begin{align*}
\vmax \le \submod(\Odual) + |\lin(\Odual)| &= -\lin(\Odual) - \obj(O) - \lin(\Odual) \\
&\le -2\lin(\Odual).
\end{align*}
Therefore, by Theorem \ref{thm:main}, the base $S$ of $\Mdual$ returned by the algorithm satisfies:
\begin{equation*}
\label{eq:supmod1}
\submod(S) + \ell(S) \ge (1 - e^{-1})\submod(\Odual) + \ell(\Odual) + O(\epsilon) \cdot \ell(\Odual)
\end{equation*}
Finally, since $\obj$ is supermodular with curvature at most $c$, Lemma \ref{lem:supmod-curv} implies that for all $A \subseteq \ground$,
$$\textstyle -\lin(A) = -\sum_{j \in A}\obj_{\emptyset}(j) \le \frac{1}{1-c}f(\ground \setminus A).$$
Thus, we have
\begin{align*}
\obj(\ground \setminus S) 
&= -\submod(S) - \lin(S) \\
&\le -\left(1-e^{-1}\right)\submod(\Odual) - \lin(\Odual) - O(\epsilon) \cdot \ell(\Odual) \\
&= \left(1-e^{-1}\right)\obj(O) - \left(e^{-1} + O(\epsilon)\right)\lin(\Odual) \\
&\le \left(1 - e^{-1}\right)\obj(O) + \left(e^{-1} + O(\epsilon)\right)\cdot\textstyle \frac{1}{1-c}\cdot \obj(O) \\
&= \left(1 + \textstyle\frac{c}{1-c}\cdot e^{-1} + \frac{1}{1-c}\cdot O(\epsilon)\right)\obj(O).
\end{align*}
\end{proof}
We note that because the error term depends on $\frac{1}{1-c}$, our result requires that $c$ is bounded away from $1$ by a constant.

\section{Inapproximability Results}
\label{sec:inappr-results}

We now show that our approximation guarantees are the best achievable using only a polynomial number of function evaluations, even in the special case that $\M$ is a uniform matroid (i.e., a cardinality constraint).  Nemhauser and Wolsey \cite{Nemhauser1978} considered the problem of finding a set $S$ of cardinality at most $r$ that maximizes a monotone submodular function. They give a class of functions for which obtaining a $(1 - e^{-1} + \epsilon)$-approximation for any constant $\epsilon > 0$ requires a super-polynomial number of function evaluations. Our analysis uses the following additional property, satisfied by every function $f$ in their class: let $p = \max_{i \in \ground}f_{\emptyset}(i)$, and let $O$ be a set of size $r$ on which $f$ takes its maximum value.  Then, $f(O) = rp$.
\begin{theorem}
\label{lem:submod-inapprox}
For any constant $\delta > 0$ and $c \in (0,1)$, there is no $(1 - c e^{-1} + \delta)$-approximation algorithm for the problem $\max \{ \hat{f}(S) : |S| \le r \}$, where $\hat{f}$ is a monotone increasing submodular function with curvature at most $c$, that evaluates $\hat{f}$ on only a polynomial number of sets.
\end{theorem}
\begin{proof}
Let $f$ be a function in the family given by \cite{Nemhauser1978} for the cardinality constraint $r$, and let $O$ be a set of size $r$ on which $f$ takes its maximum value.  Consider the function: \[\hat{f}(A) = f(A) + \frac{1 - c}{c}\cdot |A|\cdot p.\]   In Lemma \ref{lem:submod-inapprox-curv} of the appendix, we show that $\hat{f}$ is monotone increasing, submodular, and nonnegative with curvature at most $c$.

We consider the problem $\max\{\hat{f}(S) : |S| \le r\}$.  Let $\alpha = (1 - ce^{-1} + \delta)$, and suppose that some algorithm returns a solution $S$ satisfying $\hat{f}(S) \ge \alpha \cdot \hat{f}(O)$, evaluating $\hat{f}$ on only a polynomial number of sets.  Because $f$ is monotone increasing, we assume without loss of generality that $|S| = r$. Then,
\begin{equation*}
f(S) + \frac{1 - c}{c}\cdot rp \ge \alpha \cdot f(O) + \alpha \cdot \frac{1 - c}{c}\cdot rp
= \alpha \cdot f(O) + \alpha \cdot \frac{1 - c}{c}\cdot f(O) = \frac{1}{c}\cdot \alpha \cdot f(O),
\end{equation*}
and so
\begin{align*}
f(S) &\ge \frac{1}{c}\cdot \alpha \cdot f(O) - \frac{1-c}{c}\cdot rp \\
&= \left(\frac{1}{c}\cdot \alpha - \frac{1-c}{c}\right)\cdot f(O) \\
 &= \left(\frac{1}{c} - e^{-1} + \frac{\delta}{c} - \frac{1-c}{c}\right)\cdot f(O) \\ 
&= \left(1 - e^{-1} + \frac{\delta}{c}\right)\cdot f(O). 
\end{align*}
Because each evaluation of $\hat{f}$ requires only a single evaluation of $f$, this contradicts the negative result of \cite{Nemhauser1978}.
\end{proof}

\begin{theorem}
\label{lem:supmod-inapprox}
For any constant $\delta > 0$ and $c \in (0,1)$, there is no $(1 + \frac{c}{1-c}e^{-1} - \delta)$-approximation algorithm for the problem $\min \{ \hat{f}(S) : |S| \le r \}$, where $\hat{f}$ is a monotone decreasing supermodular function with curvature at most $c$, that evaluates $\hat{f}$ on only a polynomial number of sets.
\end{theorem}
\begin{proof}
Again, let $f$ be a function in the family given by \cite{Nemhauser1978} for the cardinality constraint $r$.  Let $O$ be a set of size $r$ on which $f$ takes its maximum value, and recall that $f(O) = rp$, where $p = \max_{i \in \ground} f_{\emptyset}(i)$.   Consider the function: \[\hat{f}(A) = \frac{p}{c}\cdot |\ground \setminus A| - f(\ground \setminus A).\]  In Lemma \ref{lem:supmod-inapprox-curv} of the appendix we show that $\hat{f}$ is monotone decreasing, supermodular, and nonnegative with curvature at most $c$.   

We consider the problem $\min\{\hat{f}(A) : |A| \le n - r\}$.  Let $\alpha = (1 + \frac{c}{1-c}e^{-1} - \delta)$, and suppose that some algorithm returns a solution $A$, satisfying $\hat{f}(A) \le \alpha \cdot \hat{f}(\ground \setminus O)$, evaluating $\hat{f}$ on only a polynomial number of sets.  

Suppose we run the algorithm for minimizing $\hat{f}$ and return the set $S = \ground\setminus A$.  Because $\hat{f}$ is monotone decreasing, we assume without loss of generality that $|A| = n - r$ and so $|S| = r$.  Then,
\begin{equation*}
\frac{rp}{c} - f(S) \le \alpha\left(\frac{rp}{c} - f(O)\right) 
= \alpha\left(\frac{f(O)}{c} - f(O)\right) = \alpha \cdot \frac{1-c}{c} \cdot f(O),
\end{equation*}
and so
\begin{align*}
 f(S) &\ge \frac{rp}{c} - \alpha \cdot \frac{1 - c}{c}\cdot f(O) \\
&= \left(\frac{1}{c} - \alpha \cdot \frac{1 - c}{c}\right)\cdot f(O) \\
&= \frac{1 - (1 - c) - ce^{-1} + (1-c)\delta}{c} \cdot f(O) \\
&= \left(1 - e^{-1} + \frac{1-c}{c}\cdot \delta\right)\cdot f(O).
\end{align*}
Again, since each evaluation of $\hat{f}$ requires only one evaluation of $f$, this contradicts the negative result of \cite{Nemhauser1978}.
\end{proof}

\section{Optimizing Monotone Nonnegative Functions of Bounded Curvature}
Now we consider the problem of maximizing (respectively, minimizing) an arbitrary monotone increasing (respectively, monotone decreasing) nonnegative  function $f$ of bounded curvature subject to a single matroid constraint.  We do not require that $f$ be supermodular or submodular, but only that there is a bound on the following generalized notion of curvature.

Let $f$ be an arbitrary monotone increasing or monotone decreasing function.  We define the curvature $c$ of $f$ as
\begin{equation}
\label{eq:gen-curvature}
c = 1 - \min_{j \in \ground} \min_{S,T \subseteq \ground \setminus \{j\}}\frac{f_S(j)}{f_T(j)}.
\end{equation}
Note that in the case that $f$ is either monotone increasing and submodular or monotone decreasing and supermodular, the minimum of $\frac{f_S(j)}{f_T(j)}$ over $S$ and $T$ is attained when $S = \ground - j$ and $T = \emptyset$.  Thus,  \eqref{eq:gen-curvature} agrees with the standard definition of curvature given in \eqref{eq:curvature}.  Moreover, if monotonically increasing $f$ has curvature at most $c$ for some $c \in [0,1]$, then for all $j \in \ground$ we have 
\begin{equation}
\label{eq:gen-curvature-2}
(1-c)f_B(j) \le f_A(j),
\end{equation}
for any $j \in \ground$, and $A,B \subseteq \ground \setminus \{j\}$.    
Analogously,  for monotonically decreasing function $f$ we have
\begin{equation}
\label{eq:gen-curvature-3}
(1-c)f_B(j) \ge f_A(j),
\end{equation}
for any $j \in \ground$, and $A,B \subseteq \ground \setminus \{j\}$.
As with the standard notion of curvature for submodular and supermodular functions, when $c = 0$, we find that \eqref{eq:gen-curvature-2} and  \eqref{eq:gen-curvature-3} require $f$ to be a linear function, while when $c = 1$, they require only that $f$ is monotone increasing or monotone decreasing, respectively.

First, we consider the case in which we wish to maximize a monotone increasing function $\obj$ subject to a matroid constraint $\M = (\ground,\I)$.  Suppose that we run the standard greedy algorithm, which at each step adds to the current solution $S$ the element $e$ yielding the largest marginal gain in $\obj$, subject to the constraint $S + e \in \I$.
\begin{theorem}
\label{thm:arbitrary-max}
Suppose that $\obj$ is a monotone increasing function with curvature at most $c \in [0,1]$, and $\M$ is a matroid.  Let $S \in \bases{\M}$ be the base produced by the standard greedy maximization algorithm on $\obj$ and $\M$, and let $O \in \bases{\M}$ be any base of $\M$.  Then, 
\[\obj(S) \ge (1 -c)\obj(O).\]
\end{theorem}

\begin{proof}
Let $r$ be rank of $\M$.  Let $s_i$ be the $i$th element picked by the greedy algorithm, and let $S_i$ be the set containing the first $i$ elements picked by the greedy algorithm.  We use the bijection guaranteed by Lemma \ref{lem:brualdi} to order the elements of $o_i$ of $O$ so that $S - s_i  + o_i \in \I$ for all $i \in [r]$, and let $O_i = \{o_j : j \le i\}$ be the set containing the first $i$ elements of $O$ in this ordering.  Then, 
\begin{align*}
(1-c)f(O) &= (1-c)f(\emptyset) + (1-c)\sum_{i = 1}^rf_{O_{i-1}}(o_i) \\
&\le f(\emptyset) + \sum_{i = 1}^rf_{S_{i-1}}(o_i) \\
&\le f(\emptyset) + \sum_{i = 1}^rf_{S_{i-1}}(s_i) \\
&= f(S).
\end{align*}
The first inequality follows from \eqref{eq:gen-curvature-2} and $f(\emptyset) \ge 0$.  The last inequality is due to the fact that $S_{i-1} + o_i \in \I$ but $s_i$ was chosen by the greedy maximization algorithm in the $i$th round.
\end{proof}

Similarly, we can consider the problem of finding a base of $\M$ that minimizes $f$.  In this setting, we again employ a greedy algorithm, but at each step choose the element $e$ yielding the \emph{smallest} marginal gain in $\obj$, terminating only when no element can be added to the current solution.  We call this algorithm the standard greedy minimization algorithm.

\begin{theorem}
\label{thm:arbitrary-min}
Suppose that $\obj$ is a monotone increasing function with curvature at most $c \in [0,1]$ and $\M$ is a matroid.  Let $S \in \bases{\M}$ be the base produced by the standard greedy minimization algorithm on $\obj$ and $\M$, and let $O \in \bases{\M}$ be any base of $\M$.  Then,
\[
\obj(S) \le \frac{1}{1-c}\obj(O).
\]
\end{theorem}
\begin{proof}
Let $r, S_i, s_i, O_i,$ and $o_i$ be defined as in the proof of Theorem \ref{thm:arbitrary-max}.  Then, 
\begin{align*}
f(O) &= f(\emptyset) + \sum_{i=1}^rf_{O_{i-1}}(o_i) \\
&\ge  (1-c)f(\emptyset) + (1-c)\sum_{i = 1}^rf_{S_{i-1}}(o_i) \\
&\ge  (1-c)f(\emptyset) + \sum_{i = 1}^rf_{S_{i-1}}(s_i) \\
&= (1-c)f(S).
\end{align*}
As in the proof of Theorem \ref{thm:arbitrary-max}, the first inequality follows from \eqref{eq:gen-curvature-2} and $f(\emptyset) 
\ge 0$.  The last inequality is due to the fact that $S_{i-1} + o_i \in \I$ but $s_i$ was chosen by the greedy minimization algorithm in the $i$th round.
\end{proof}

Now, we consider the case in which $f$ is a monotone decreasing function.  For any function $f : 2^\ground \to \posreals$, we define the function $\fdual : 2^\ground \to \posreals$ by $\fdual(S) = f(\ground \setminus S)$ for all $S \subseteq \ground$.  Then, since $f$ is monotone decreasing, $\fdual$ is monotone increasing.  Moreover, the next lemma shows that the curvature of $\fdual$ is the same as that of $f$.

\begin{lemma}
\label{lem:increasing-decreasing-curvature}
Let $f$ be a monotone decreasing function with curvature at most $c \in [0,1]$, and define $\fdual(S) = f(\ground \setminus S)$ for all $S \subseteq \ground$.  Then, $\fdual$ also has curvature at most $c$.
\end{lemma}
\begin{proof}
From the definition of $\fdual$, we have:
\[
\fdual_A(j) = f(\ground \setminus (A + j)) - f(\ground \setminus A) = -f_{\ground\setminus (A + j)}(j),
\]
for any $A \subseteq \ground$ and $j \in \ground$.  Consider any $j \in \ground$ and $S,T \subseteq \ground \setminus \{j\}$.   Since $\obj$ is monotone decreasing with curvature at most $c$, \eqref{eq:gen-curvature-3} implies
\begin{equation*} 
\fdual_{S}(j)
= -f_{\ground \setminus (S + j)}(j)
\ge
-(1-c)f_{\ground \setminus (T + j)}(j)
= (1-c)\fdual_T(j).
\end{equation*}
Thus, $\frac{\fdual_S(j)}{\fdual_T(j)} \ge (1-c)$ for all $j \in \ground$ and $S,T \subseteq \ground \setminus \{j\}$.
\end{proof}

Given a matroid $\M$, we consider the problem of finding a base of $\M$ minimizing $f$.  This problem is equivalent to finding a base of the dual matroid $\Mdual$ that minimizes $\fdual$.  Similarly, the problem of finding a base of $\M$ that maximizes $f$ can be reduced to that of finding a base of $\Mdual$ that maximizes $\fdual$.  Since $\fdual$ is monotone increasing with curvature no more than $f$, we obtain the following corollaries, which show how to employ the standard greedy algorithm to optimize monotone decreasing functions.
\begin{corollary}
\label{thm:arbitrary-max-decreasing}
Suppose that $f$ is a monotone decreasing function with curvature at most $c \in [0,1]$ and $\M$ is a matroid.  Let $\Sdual \in \bases{\Mdual}$ be the base of $\Mdual$ produced by running the standard greedy maximization algorithm on $\fdual$ and $\Mdual$.  Let $O \in \bases{\M}$ be any base of $\M$, $\Odual = \ground \setminus O$, and $S = \ground \setminus \Sdual \in \bases{\M}$.  Then,
\[
f(S) = \fdual(\Sdual) \ge (1-c)\fdual(\Odual) = (1-c)f(O).
\]
\end{corollary}
\begin{corollary}
\label{thm:arbitrary-min-decreasing}
Suppose that $f$ is a monotone decreasing function with curvature at most $c \in [0,1]$ and $\M$ is a matroid.  Let $\Sdual \in \bases{\Mdual}$ be the base of $\Mdual$ produced by running the standard greedy minimization algorithm on $\fdual$ and $\Mdual$.  Let $O \in \bases{\M}$ be any base of $\M$, $\Odual = \ground \setminus O$, and $S = \ground \setminus \Sdual \in \bases{\M}$.  Then,
\[
f(S) = \fdual(\Sdual) \le \frac{1}{1-c}\fdual(\Odual) = \frac{1}{1-c}f(O).
\]
\end{corollary}

The approximation factors of $1-c$ and $1/(1-c)$ for maximization and minimization respectively are best possible, given curvature $c$.
The hardness result for minimization follows from \cite{IJB}, where it is shown that no algorithm using polynomially many value queries can achieve an approximation factor of $\rho(n,\epsilon) = \frac{n^{1/2 - \epsilon}}{1 + (n^{1/2 - \epsilon} - 1)(1 - c)}$ for the problem $\min \{ f(S): |S| \ge r \}$, where $f$ is monotone increasing (even submodular) of curvature $c$. This implies that no $1/(1-c+\epsilon)$-approximation for a constant $\epsilon>0$ is possible for this problem. Next, we prove the hardness result for maximization; this proof is based on known hardness constructions for maximization of XOS functions \cite{DobzinskiS06,MirrokniSV08}.

\begin{theorem}
\label{thm:curvature-max-hardness}
For any constant $c \in (0,1)$ and $\epsilon>0$, there is no $(1-c+\epsilon)$-approximation using polynomially many queries for the problem $\max \{f(S): |S| \leq k \}$ where $f$ is monotone increasing of curvature $c$.
\end{theorem}

\begin{proof}
Fix $c \in (0,1)$, $|X| = n$ and let $O \subseteq X$ be a random subset of size $k = n^{2/3}$ (assume $k$ is an integer). We define the following function:
\begin{itemize}
\item $ f(S) = \max \{ |S \cap O|, (1-c) |S| + c n^{1/3}, \min \{|S|, n^{1/3}\} \}$.
\end{itemize}
The marginal values of $f$ are always between $1-c$ and $1$; therefore, its curvature is $c$.
Consider a query $Q$.
For $|Q| \le n^{1/3}$, we have $f(Q) = |Q|$ in any case. For $|Q| > n^{1/3}$, we have $f(Q) = (1-c) |Q| + c n^{1/3}$, unless
$|Q \cap O| > (1-c) |Q| + c n^{1/3}$. Since $O$ is a random $\frac{1}{n^{1/3}}$-fraction of the ground set, we have $\expect{|Q \cap O|} = \frac{1}{n^{1/3}} |Q| \geq 1$.
Therefore, by the Chernoff bound, $\Pr[|Q \cap O| > (1-c) |Q| + c n^{1/3}]$ is exponentially small in $n^{1/3}$ (with respect to the choice of $O$).
Furthermore, as long as $|Q \cap O| \leq (1-c) |Q| + c n^{1/3}$, $f(Q)$ depends only on $|Q|$ and hence the algorithm does not learn anything about the identity of the optimal set $O$. Unless the algorithm queries an exponential number of sets, it will never find a set $S$ satisfying $|S \cap O| > (1-c) |S| + c n^{1/3}$ and hence the value of the returned solution is $f(S) \leq (1-c) |S| + c n^{1/3} \leq (1-c) n^{2/3} + c n^{1/3}$ with high probability. On the other hand, the optimum is $f(O) = |O| = n^{2/3}$. Therefore, the algorithm cannot achieve a better than $(1-c+ o(1))$-approximation.
\end{proof}

Therefore, the approximation factors in Theorems~\ref{thm:arbitrary-max} and \ref{thm:arbitrary-min} are optimal.
Combining these inapproximability results with Lemma \ref{lem:increasing-decreasing-curvature} we derive similar inapproximability results showing the optimality of Corollary~\ref{thm:arbitrary-max-decreasing} and \ref{thm:arbitrary-min-decreasing}.

\section{Application: the Column-Subset Selection Problem}\label{application}

Let $A$ be an $m \times n$ real matrix. We denote the columns of $A$ by $\bc_1,\ldots,\bc_n$. I.e., for $\bx \in \reals^n$,
$A \bx = \sum x_i \bc_i$. The (squared) Frobenius norm of $A$ is defined as $$\|A\|^{2}_F = \sum_{i,j} a_{ij}^2 = \sum_{i=1}^{n} \|\bc_i\|^2,$$ where here, and throughout this section, we use $\|\!\cdot\!\|$ to denote the standard, $\ell_{2}$ vector norm.

For a matrix $A$ with independent columns, the {\em condition number} is defined as
$$ \kappa(A) = \frac{\sup_{\|\bx\|=1} \|A \bx\|}{\inf_{\|\bx\|=1} \|A \bx\|}.$$  If the columns of $A$ are dependent, then $\kappa(A) = \infty$ (there is a nonzero vector $\bx$ such that $A\bx = 0$).

Given a matrix $A$ with columns $\bc_1,\ldots,\bc_n$, and a subset $S \subseteq [n]$, we denote by $$\proj_S(\bx) = \mbox{argmin}_{\by \in \spn(\{\bc_i: i \in S\})} \|\bx - \by\|$$ the projection of $\bx$ onto the subspace spanned by the respective columns of $A$.
Given $S \subseteq [n]$, it is easy to see that the matrix $A(S)$ with columns spanned by $\{ \bc_i: i \in S\}$ that is closest to $A$ in squared Frobenius norm is $A(S) = (\proj_S(\bc_1), \proj_S(\bc_2), \ldots, \proj_S(\bc_n))$. The distance between  two matrices is thus
$$ \|A - A(S)\|^{2}_F= \sum_{i=1}^{n} \| \bc_i - \proj_S(\bc_i) \|^2.$$
We define $f^A:2^{[n]} \rightarrow \reals$ to be this quantity as a function of $S$:
\begin{align*}
f^A(S) &= \sum_{i=1}^{n} \| \bc_i - \proj_S(\bc_i) \|^2 \\
&= \sum_{i=1}^{n} ( \|\bc_i\|^2 - \| \proj_S(\bc_i) \|^2),
\end{align*}
where the final equality follows from the fact that $\proj_{S}(\bc_i)$ and $\bc_{i} - \proj_{S}(\bc_{i})$ are orthogonal.

Given a matrix $A \in \reals^{m \times n}$ and an integer $k$, the \emph{column-subset selection problem} (CSSP) is to select a subset $S$ of $k$ columns of $A$ so as to minimize $f^A(S)$.
It follows from standard properties of projection that $f^A$ is non-increasing, and so CSSP is a special case of non-increasing minimization subject to a cardinality constraint.  We now show that the curvature of $f^{A}$ is related to the condition number of $A$.

\begin{lemma}
\label{lem:condition-curvature}
For any non-singular matrix $A$, the curvature $c=c({f^A})$
of the associated set function $f_A$ satisfies 
$$ \frac{1}{1-c} \leq \kappa^2(A).$$
\end{lemma}

\begin{proof}
We want to prove that for any $S$ and $i \notin S$, 
\begin{eqnarray}\label{conditionNumberInequality}
 \min_{\|\bx\|=1} \|A\bx\|^2 \leq  |f^A_S(i)| \leq \max_{\|\bx\|=1} \|A\bx\|^2.
\end{eqnarray}
This implies that by varying the set $S$, a marginal value can change by at most a factor of $\kappa^2(A)$. 
Recall that the marginal values of $f^A$ are negative, but only the ratio matters so we can consider the respective absolute values. 
The inequalities (\ref{conditionNumberInequality}) imply that 
$$ \frac{1}{1-c}=\max_{j \in \ground} \max_{S,T \subseteq \ground \setminus \{j\}}\frac{|f_S(j)|}{|f_T(j)|}   \leq \kappa^2(A).$$

We now prove the  inequalities (\ref{conditionNumberInequality}). Let $\bv_{i,S} = \bc_i - \proj_S(\bc_i)$ denote the component of $\bc_i$ orthogonal to $\spn(S)$.
We have
\begin{align*}
|f^A_S(i)| &=  f^A(S) - f^A(S+i) \\
&= \sum_{j=1}^{n} \left( \|\bc_j - \proj_S(\bc_j)\|^2 - \|\bc_j - \proj_{S+i}(\bc_j)\|^2 \right) \\
&= \sum_{j=1}^{n} \left( \|\bc_j - \proj_S(\bc_j)\|^2 - \|\bc_j - \proj_{S}(\bc_j) - \proj_{\bv_{i,S}}(\bc_j) \|^2 \right) \\
&= \sum_{j=1}^{n} \| \proj_{\bv_{i,S}}(\bc_j) \|^2
\end{align*}
because $\proj_S(\bc_j)$, $\proj_{\bv_{i,S}}(\bc_j)$ and $\bc_j - \proj_{S}(\bc_j) - \proj_{\bv_{i,S}}(\bc_j)$ are orthogonal.

Our first goal is to show that if $|f^A_S(i)|$  is large, then there is a unit vector $\bx$ such that $\|A\bx\|$ is large. In particular, let us define $p_j = (\bv_{i,S} \cdot \bc_j)/\|\bv_{i,S} \|$ and $x_j = p_j / \sqrt{ \sum_{\ell=1}^{n} p_\ell^2}$. We have $\|\bx\|^2 = \sum_{j=1}^{n} x_j^2 = 1$. Multiplying by matrix $A$, we obtain $A\bx = \sum_{j=1}^{n} x_j \bc_j$. We can estimate $\|A \bx\|$ as follows:
\begin{align*}
 \bv_{i,S} \cdot (A \bx) &= \bv_{i,S} \cdot \sum_{j=1}^{n} x_j \bc_j  \\
&= \sum_{j=1}^{n} \frac{p_j}{\sqrt{\sum_{\ell=1}^{n} p_\ell^2}} (\bv_{i,S} \cdot \bc_j) 
\\ 
&= \sum_{j=1}^{n} \frac{p_j^2}{\sqrt{\sum_{\ell=1}^{n} p_\ell^2}} \|\bv_{i,S}\|  \\
&= \| \bv_{i,S} \| \sqrt{ \sum_{j=1}^{n} p_j^2}.
\end{align*}
By the Cauchy-Schwartz inequality, this implies that $\| A\bx \| \geq \sqrt{ \sum_{j=1}^{n} p_j^2} = \sqrt{ |f^A_S(i)|}$.

On the other hand, using the expression above, we have
\begin{align*}
 |f^A_S(i)| = \sum_{j=1}^{n} \| \proj_{\bv_{i,S}}(\bc_j) \|^2 \geq \| \proj_{\bv_{i,S}}(\bc_i) \|^2 = \| \bv_{i,S} \|^2
\end{align*}
since $\bv_{i,S} = \bc_i - \proj_S(\bc_i)$ is the component of $\bc_i$ orthogonal to $\spn(S)$.
We claim that if $\| \bv_{i,S} \|$ is small, then there is a unit vector $\bx'$ such that $\|A\bx'\|$ is small. To this purpose, write $\proj_{S}(\bc_i)$ as a linear combination of the vectors $\{\bc_j: j \in S \}$: $ \proj_{S}(\bc_i) = \sum_{j \in S} y_j \bc_j$. Finally, we define $y_i = -1$, and normalize to obtain $\bx' = \by / \| \by \|$.
We get the following:
\begin{align*}
\| A \bx' \| = \frac{1}{\| \by \|} \| A \by \| &= \frac{1}{\| \by \|}  \| \sum_{j=1}^{n} y_j \bc_j \| \\
&= \frac{1}{\| \by \|} \| \proj_{S}(\bc_i) - \bc_i \| \\
&= \frac{1}{\| \by \|} \| \bv_{i,S} \|.
\end{align*}
Since $\| \by \| \geq 1$, and $\| \bv_{i,S} \| \leq \sqrt{|f^A_{S}(i)|}$, we obtain $\|A \bx' \| \leq \sqrt{|f^A_{S}(i)|}$.

In summary, we have given two unit vectors $\bx, \bx'$ with $\|A \bx \| \geq \sqrt{|f^A_{S}(i)|}$ and $\|A \bx' \| \leq \sqrt{|f^A_{S}(i)|}$.
This proves that $ \min_{\|\bx\|=1} \|A\bx\|^2 \leq  |f^A_S(i)| \leq \max_{\|\bx\|=1} \|A\bx\|^2.$
\end{proof}

By Corollary \ref{thm:arbitrary-min-decreasing}, the standard greedy minimization algorithm is then a $\kappa^2(A)$-approximation for the column-subset selection problem.
The following lemma shows that Lemma~\ref{lem:condition-curvature} is asymptotically tight.

\begin{lemma}
There exists a matrix $A$ with condition number $\kappa$ for which the associated function $f^{A}$ has curvature $1 - O(1/\kappa^{2})$.
\end{lemma}

Let us denote by $\dist_S$ the distance from $\bx$ to the subspace spanned by the columns corresponding to $S$.
$$ \dist_S(\bx) = \|\bx - \proj_S(\bx)\| = \mbox{min}_{\by \in \spn(\{\bc_i: i \in S\})} \|\bx - \by\|.$$


For some $\epsilon>0$, consider $A = (\bc_1, \ldots, \bc_n)$ where $\bc_1 = \be_1$ and $\bc_j = \epsilon \be_1 + \be_j$ for $j \geq 2$.
(A similar example was used in \cite{Boutsidis2014} for a lower bound on column-subset approximation.) Here, $\be_i$ is the $i$-th canonical basis vector in $R^n$. We claim that the condition number of $A$ is $\kappa = O(\max \{1, \epsilon^2 (n-1) \})$, while the curvature of $f^A$ is $1 - O(\frac{1}{\max \{1, \epsilon^4 (n-1)^2 \}}) = 1 - O(\frac{1}{\kappa^2})$.

To bound the condition number, consider a unit vector $\bx$. We have
$$ A \bx = (x_1 + \epsilon \sum_{i=2}^{n} x_i, x_2, x_3, \ldots, x_n) $$
and
\begin{align*} \|A \bx \|^2 &= (x_1 + \epsilon \sum_{i=2}^{n} x_i)^2 + \sum_{i=2}^{n} x_i^2 \\ &= \| \bx \|^2 + 2 \epsilon x_1 \sum_{i=2}^{n} x_i + \epsilon^2 (\sum_{i=2}^{n} x_i)^2.
\end{align*}
We need a lower bound and an upper bound on $\|A \bx \|$, assuming that $\| \bx\| = 1$. On the one hand, we have
\begin{align*}
\| A \bx \|^2 &\leq 1 + (x_1 + \epsilon \sum_{i=2}^{n} x_i)^2 \leq 1 + (1 + \epsilon \sqrt{n-1})^2 \\
&= O( \max \{1, \epsilon^2 (n-1) \}).
\end{align*}
On the other hand, to get a lower bound: if $x_1 \leq 1/2$, then $\| A \bx \|^2 \geq \sum_{i=2}^{n} x_i^2 = 1 - x_1^2 \geq \frac34$. If $x_1 > 1/2$, then either $\sum_{i=2}^{n} |x_i| \leq \frac{1}{4 \epsilon}$, in which case $$\| A \bx \|^2 \geq (x_1 + \epsilon \sum_{i=2}^{n} x_i)^2 \geq \frac{1}{16},$$ or $\sum_{i=2}^{n} |x_i| > \frac{1}{4 \epsilon}$ in which case by convexity we get $$\sum_{i=2}^{n} x_i^2 \geq \frac{1}{16 \epsilon^2 (n-1)}.$$ So, in all cases $\| A \bx \|^2 = \Omega(1 / \max \{1,\epsilon^2 (n-1)\})$. This means that the condition number of $A$ is $\kappa = O(\max \{1, \epsilon^2(n-1) \})$.

To lower-bound the curvature of $f^A$, consider the first column $\bc_1$ and let us estimate $f^A_\emptyset(1)$ and $f^A_{[n]\setminus \{1\}}(1)$. We have
$$ f^A_\emptyset(1) = \| \bc_1 \|^2 + \sum_{j=2}^{n} \| \proj_1(\bc_j) \|^2 = 1 + \epsilon^2 (n-1).$$
On the other side,
$$ f^A_{[n] \setminus \{1\}}(1) = \| \bc_1 - \proj_{[n] \setminus \{1\}} \bc_1 \|^2 = (\dist_{[n] \setminus \{1\}}(\bc_1))^2. $$
We exhibit a linear combination of the columns $\bc_2,\ldots,\bc_n$ which is close to $\bc_1$: Let $\by = \frac{1}{\epsilon(n-1)} \sum_{j=2}^{n} \bc_j$. 
We obtain
$$ \dist_{[n] \setminus \{1\}} (\bc_1) \leq \| \bc_1 - \by \| 
 = \frac{1}{\epsilon (n-1)} \| (0,-1,-1,\ldots,-1) \|
= \frac{1}{\epsilon \sqrt{n-1}}.$$
Alternatively, we can also pick $\by = 0$ which shows that $\dist_{[n] \setminus \{1\}}(\bc_1) \leq \| \bc_1 \| = 1$. So we have
$$ f^A_{[n] \setminus \{1\}}(1) = (\dist_{[n] \setminus \{1\}}(\bc_1))^2  \leq \min \left\{1, \frac{1}{\epsilon^2 (n-1)} \right\} 
 = \frac{1}{\max \{1, \epsilon^2 (n-1) \}}. $$
We conclude that the curvature of $f^A$ is at least $$1 - \frac{1}{\max \{1, \epsilon^4 (n-1)^2\}} = 1 - O\left(\frac{1}{\kappa^2}\right).$$

\paragraph{Acknowledgment}
We thank Christos Boutsidis for suggesting a connection between curvature and condition number.

\bibliographystyle{plain}
\bibliography{names,Curvature-full}

\begin{thebibliography}{10}

\bibitem{Ageev2004}
Alexander Ageev and Maxim Sviridenko.
\newblock Pipage rounding: A new method of constructing algorithms with proven
  performance guarantee.
\newblock {\em J. Combinatorial Optimization}, 8(3):307--328, 2004.

\bibitem{Boutsidis2014}
Christos Boutsidis, Petros Drineas, and Malik Magdon-Ismail.
\newblock Near-optimal column-based matrix reconstruction.
\newblock {\em SIAM J. Comput.}, 43(2):687--717, 2014.

\bibitem{Brualdi1969}
Richard~A Brualdi.
\newblock Comments on bases in dependence structures.
\newblock {\em Bull. of the Australian Math. Soc.}, 1(02):161--167, 1969.

\bibitem{Calinescu2007}
Gruia Calinescu, Chandra Chekuri, Martin P\'{a}l, and Jan Vondr\'{a}k.
\newblock Maximizing a submodular set function subject to a matroid constraint
  (extended abstract).
\newblock In {\em Proc. 12th IPCO}, pages 182--196, 2007.

\bibitem{Calinescu2011}
Gruia Calinescu, Chandra Chekuri, Martin P\'{a}l, and Jan Vondr\'{a}k.
\newblock Maximizing a submodular set function subject to a matroid constraint.
\newblock {\em SIAM J. Comput.}, 40(6):1740--1766, 2011.

\bibitem{Chekuri2010}
Chandra Chekuri, Jan Vondr\'{a}k, and Rico Zenklusen.
\newblock Dependent randomized rounding via exchange properties of
  combinatorial structures.
\newblock In {\em Proc. 51st FOCS}, pages 575--584, 2010.

\bibitem{Chekuri2011a}
Chandra Chekuri, Jan Vondr\'{a}k, and Rico Zenklusen.
\newblock {Multi-budgeted matchings and matroid intersection via dependent
  rounding}.
\newblock In {\em Proc. 22nd SODA}, pages 1080--1097, 2011.

\bibitem{Chekuri2011}
Chandra Chekuri, Jan Vondr\'{a}k, and Rico Zenklusen.
\newblock {Submodular function maximization via the multilinear relaxation and
  contention resolution schemes}.
\newblock In {\em Proc. 43rd STOC}, pages 783--792, 2011.

\bibitem{Conforti1984}
Michele Conforti and G\'{e}rard Cornu\'{e}jols.
\newblock Submodular set functions, matroids and the greedy algorithm: Tight
  worst-case bounds and some generalizations of the rado-edmonds theorem.
\newblock {\em Discrete Applied Mathematics}, 7(3):251--274, 1984.

\bibitem{Deshpande2010}
Amit Deshpande and Luis Rademacher.
\newblock Efficient volume sampling for row/column subset selection.
\newblock In {\em Proc. 51st FOCS}, pages 329--338, 2010.

\bibitem{Deshpande2006}
Amit Deshpande and Santosh Vempala.
\newblock Adaptive sampling and fast low-rank matrix approximation.
\newblock In {\em Proc. 9th APPROX}, pages 292--303. Springer, 2006.

\bibitem{DNS05}
S.~Dobzinski, N.~Nisan, and M.~Schapira.
\newblock Approximation algorithms for combinatorial auctions with
  complement-free bidders.
\newblock In {\em Proc. 37th STOC}, pages 610--618, 2005.

\bibitem{DobzinskiS06}
Shahar Dobzinski and Michael Schapira.
\newblock An improved approximation algorithm for combinatorial auctions with
  submodular bidders.
\newblock In {\em Proceedings of the Seventeenth Annual {ACM-SIAM} Symposium on
  Discrete Algorithms, {SODA} 2006, Miami, Florida, USA, January 22-26, 2006},
  pages 1064--1073, 2006.

\bibitem{Dobzinski2012}
Shahar Dobzinski and Jan Vondr\'{a}k.
\newblock {From query complexity to computational complexity}.
\newblock In {\em Proc. 44th STOC}, pages 1107--1116, 2012.

\bibitem{Edmonds1971}
Jack Edmonds.
\newblock {Matroids and the greedy algorithm}.
\newblock {\em Mathematical Programming}, 1(1):127--136, 1971.

\bibitem{Feige1998}
Uriel Feige.
\newblock {A threshold of ln n for approximating set cover}.
\newblock {\em J. ACM}, 45:634--652, 1998.

\bibitem{Filmus2012}
Yuval Filmus and Justin Ward.
\newblock A tight combinatorial algorithm for submodular maximization subject
  to a matroid constraint.
\newblock In {\em Proc. 53nd FOCS}, 2012.

\bibitem{Filmus2014}
Yuval Filmus and Justin Ward.
\newblock Monotone submodular maximization over a matroid via non-oblivious
  local search.
\newblock {\em SIAM J. Comput.}, 43(2):514--542, 2014.

\bibitem{Fisher1978}
M.L. Fisher, G.L. Nemhauser, and L.A. Wolsey.
\newblock An analysis of approximations for maximizing submodular set
  functions---{II}.
\newblock {\em Mathematical Programming Studies}, 8:73--87, 1978.

\bibitem{Ilev2001}
Victor~P. Il'ev.
\newblock {An approximation guarantee of the greedy descent algorithm for
  minimizing a supermodular set function}.
\newblock {\em Discrete Applied Mathematics}, 114(1-3):131--146, October 2001.

\bibitem{IJB}
R.~Iyer, S.~Jegelka, and J.~Bilmes.
\newblock Curvature and optimal algorithms for learning and minimizing
  submodular functions.
\newblock In {\em In Neural Information Processing Society (NIPS), Lake Tahoe,
  CA,}, 2013.

\bibitem{Kelmans1983}
A.K. Kelmans.
\newblock {Multiplicative submodularity of a matrix's principal minor as a
  function of the set of its rows}.
\newblock {\em Discrete Mathematics}, 44(1):113--116, 1983.

\bibitem{KKT03}
David Kempe, Jon~M. Kleinberg, and {\'E}va Tardos.
\newblock Maximizing the spread of influence through a social network.
\newblock In {\em Proc. 9th KDD}, pages 137--146, 2003.

\bibitem{Lee1996}
C.W. Ko, Jon Lee, and Maurice Queyranne.
\newblock An exact algorithm for maximum entropy sampling.
\newblock {\em Operations Research}, 43(4):684--691, 1996.

\bibitem{KG11}
A.~Krause and C.~Guestrin.
\newblock Submodularity and its applications in optimized information
  gathering.
\newblock {\em ACM Trans. on Intelligent Systems and Technology}, 2(4):32,
  2011.

\bibitem{KGGK06}
A.~Krause, C.~Guestrin, A.~Gupta, and J.~Kleinberg.
\newblock Near-optimal sensor placements: maximizing information while
  minimizing communication cost.
\newblock In {\em Proc. 5th IPSN}, pages 2--10, 2006.

\bibitem{KSG08}
A.~Krause, A.~Singh, and C.~Guestrin.
\newblock Near-optimal sensor placements in gaussian processes: Theory,
  efficient algorithms and empirical studies.
\newblock {\em J. Machine Learning Research}, 9:235--284, 2008.

\bibitem{KRGG09}
Andreas Krause, Ram Rajagopal, Anupam Gupta, and Carlos Guestrin.
\newblock Simultaneous placement and scheduling of sensors.
\newblock In {\em Proc. 8th IPSN}, pages 181--192, 2009.

\bibitem{Lee2002}
Jon Lee.
\newblock Maximum entropy sampling.
\newblock {\em Encyclopedia of Environmetrics}, 3:1229--1234, 2002.

\bibitem{LLN06}
B.~Lehmann, D.~J. Lehmann, and N.~Nisan.
\newblock Combinatorial auctions with decreasing marginal utilities.
\newblock {\em Games and Economic Behavior}, 55:1884--1899, 2006.

\bibitem{MirrokniSV08}
Vahab~S. Mirrokni, Michael Schapira, and Jan Vondr{\'{a}}k.
\newblock Tight information-theoretic lower bounds for welfare maximization in
  combinatorial auctions.
\newblock In {\em Proceedings 9th {ACM} Conference on Electronic Commerce
  (EC-2008), Chicago, IL, USA, June 8-12, 2008}, pages 70--77, 2008.

\bibitem{Nemhauser1978}
G.L. Nemhauser and L.A. Wolsey.
\newblock Best algorithms for approximating the maximum of a submodular set
  function.
\newblock {\em Mathematics of Operations Research}, 3(3):177--188, 1978.

\bibitem{Nemhauser1978a}
G.L. Nemhauser, L.A. Wolsey, and M.L. Fisher.
\newblock An analysis of approximations for maximizing submodular set
  functions---{I}.
\newblock {\em Mathematical Programming}, 14(1):265--294, 1978.

\bibitem{Schrijver2003}
Alexander Schrijver.
\newblock {\em {Combinatorial Optimization: Polyhedra and Efficiency}}.
\newblock Springer, 2003.

\bibitem{Vondrak2008}
Jan Vondr\'{a}k.
\newblock {Optimal approximation for the submodular welfare problem in the
  value oracle model}.
\newblock In {\em Proc. 40th STOC}, pages 67--74, 2008.

\bibitem{Vondrak2009}
Jan Vondr\'{a}k.
\newblock Symmetry and approximability of submodular maximization problems.
\newblock In {\em Proc. 50th FOCS}, pages 651--670, 2009.

\bibitem{Vondrak2010}
Jan Vondr\'{a}k.
\newblock {Submodularity and curvature: the optimal algorithm}.
\newblock In {\em RIMS Kokyuroku Bessatsu}, volume B23, pages 253--266, Kyoto,
  2010.

\end{thebibliography}

\appendix

\section{Proofs and Claims Omitted from the Main Body}
\label{sec:omitt-proofs-claims}

\submodcurv*
\begin{proof}
We order the elements of $\ground$ arbitrarily, and let $A_{j}$ be the set containing all those elements of $A$ that precede the element $j$.  Then, $\sum_{j \in A}f_{A_{j}}(j) = f(A) - f(\emptyset)$.  From \eqref{eq:curvature}, we have
\[
c \ge 1- \frac{f_{\ground-j}(j)}{f_{\emptyset}(j)}
\]
which, since $f_{\emptyset}(j) \ge 0$, is equivalent to
\[
f_{\ground-j}(j) \ge (1-c)f_{\emptyset}(j) \text{, for each }j \in A.
\]
Because $f$ is submodular, we have $f_{\emptyset}(j) \ge f_{A_{j}}(j)$ for all $j$,
and so
\begin{equation*}
\sum_{j \in A}f_{\ground - j}(j) \ge (1-c)\sum_{j \in A}f_{\emptyset}(j) \ge (1-c)\sum_{j\in A}f_{A_{j}}(j)
= (1-c)[f(A) - f(\emptyset)] \ge (1-c)f(A). 
\end{equation*}
\end{proof}

\supmodcurv*
\begin{proof}
Order $A$ arbitrarily, and let $A_{j}$ be the set of all elements in $A$ that precede element $j$, including $j$ itself.  Then, 
$\sum_{j \in A}f_{\ground \setminus A_{j}}(j) = f(\ground) - f(\ground \!\setminus\! A)$.
From \eqref{eq:curvature}, we have
\[c \ge 1 - \frac{f_{\ground - j}(j)}{f_{\emptyset}(j)},
\]
which, since $f_{\emptyset}(j) \le 0$, is equivalent to
\[
f_{\ground - j}(j) \le (1-c)f_{\emptyset}(j).
\]
Then, since $f$ is supermodular, we have $f_{\ground \setminus A_{j}}(j) \le f_{\ground - j}(j)$ for all $j \in A$,  and so
\begin{equation*}
(1 - c)\sum_{j \in A}f_{\emptyset}(j) \ge 
\sum_{j \in A}f_{\ground - j}(j) \ge
\sum_{j \in A}f_{\ground \setminus A_{j}}(j)
= f(\ground) - f(\ground \!\setminus\! A) \ge -f(\ground \!\setminus\! A).
\end{equation*}
\end{proof}

\begin{lemma}
\label{lem:submod-app}
Let $\obj : 2^{\ground} \to \posreals$ be a monotone-increasing submodular function and define $\lin(A) = \sum_{j \in A}\obj_{\ground - j}(j)$ and $\submod(A) = \obj(A) - \lin(A)$.  Then, $\submod$ is submodular,  monotone increasing, and nonnegative.
\end{lemma}
\begin{proof}
The function $\submod$ is the sum of a submodular function $\obj$ and a linear function $\lin$, and so must be submodular.  For any set $A \subseteq \ground$ and element $j \not\in A$,
\[\submod_{A}(j) = \obj_{A}(j) - \obj_{\ground - j}(j) \ge 0\]
since $\obj$ is submodular.
Thus, $\submod$ is monotone increasing.  Finally, we note that $\submod(\emptyset) = \obj(\emptyset) - \lin(\emptyset) = \obj(\emptyset) \ge 0$ and so $\submod$ must be nonnegative.
\end{proof}

\begin{lemma}
\label{lem:supmod-app}
Let $\obj : 2^{\ground} \to \posreals$ be a monotone-decreasing supermodular function and define $\lin(A) = \sum_{j \in A}\obj_{\emptyset}(j)$ and $\submod(A) = -\lin(A) - \obj(\ground\!\setminus\!A)$.    Then, $\submod$ is submodular, monotone increasing, and nonnegative.
\end{lemma}
\begin{proof}
We first show that $\submod$ is monotone-increasing.  Consider an arbitrary $A \subseteq \ground$ and $j \not\in A$.  Then,
\begin{align*}
\submod_{A}(j) &= \submod(A + j) - \submod(A) \\
&=  -\lin(A + j)  - \obj((\ground\!\setminus\!A) - j) + \lin(A) + \obj(\ground\!\setminus\!A) \\
&= -\lin(j) + \obj_{(\ground \setminus A) - j}(j) \\
&= -\obj_{\emptyset}(j) + \obj_{(\ground \setminus A) - j}(j)\\
&\ge 0,
\end{align*}
where the last line holds because $\obj$ is supermodular.  Moreover, note that $\submod(\emptyset) = -\ell(\emptyset) - \obj(\ground) = 0$, so $\submod$ is nonnegative.  

Finally, we show that $\submod$ is submodular.  Suppose $A \subseteq B$ and $j \not\in B$.  Then, $(\ground \!\setminus\! B) - j \subseteq (\ground \!\setminus\! A) - j$ and so, since $\obj$ is supermodular, $\obj_{(\ground \setminus B) - j}(j) \le \obj_{(\ground \setminus A) - j}(j)$.  Thus,
\begin{equation*}
\submod_{A}(j) = -\obj_{\emptyset}(j) + \obj_{(\ground \setminus A) - j}(j) \\
\ge -\obj_{\emptyset}(j) + \obj_{(\ground \setminus B) - j}(j) = \submod_{B}(j).
\end{equation*}
\end{proof}

\begin{lemma}
\label{lem:submod-inapprox-curv}
Let $f$ be a monotone increasing submodular function, satisfying $f_{A}(j) \le p$ for all $j, A$, and let $c \in [0,1]$.  Define $$\hat{f}(A) = f(A) + \frac{1 - c}{c}\cdot|A|\cdot p.$$  Then, $\hat{f}$ is submodular, monotone increasing, and nonnegative, and has curvature at most $c$.
\end{lemma}
\begin{proof}
Because $\hat{f}$ is the sum of a monotone increasing, nonnegative submodular function and a nonnegative linear function, and hence must be monotone increasing, nonnegative, and submodular.  Furthermore, for any $A \subseteq \ground$ and $j \not\in A$, we have $\hat{f}_{A}(j) = f_{A}(j) + \frac{1 - c}{c}p$.  Thus,
\begin{equation*}
\frac{\hat{f}_{\ground - j}(j)}{\hat{f}_{\emptyset}(j)} = \frac{f_{\ground - j}(j) + \frac{1 - c}{c}\cdot p}{f_{\emptyset}(j) + \frac{1 - c}{c}\cdot p}
\ge \frac{\frac{1 - c}{c}\cdot p}{p + \frac{1 - c}{c}\cdot p} = \frac{\frac{1-c}{c}}{\frac{1}{c}} = 1 - c,
\end{equation*}
and so $\hat{f}$ has curvature at most $c$.
\end{proof}

\begin{lemma}
\label{lem:supmod-inapprox-curv}
Let $f$ be a monotone increasing submodular function, satisfying $f_{A}(j) \le p$ for all $j, A$, and let $c \in [0,1]$.  Define:
$$\hat{f}(A) = \frac{p}{c}\cdot |\ground \!\setminus\! A| - f(\ground \!\setminus\! A).$$  Then, $\hat{f}$ is submodular, monotone increasing, and nonnegative, and has curvature at most $c$.
\end{lemma}
\begin{proof}
Because $f$ is submodular, so is $f(\ground \setminus A)$, and hence $-f(\ground \setminus A)$ is supermodular.  Thus, $\hat{f}$ is the sum of a supermodular function and a linear function and so is supermodular.  In order to see that $\hat{f}$ is decreasing, we consider the marginal $\hat{f}_{A}(j)$, which is equal to
\begin{equation*}
\frac{p}{c}\!\cdot\!|\ground\!\setminus\!(A\!+\!j)| - f(\ground\!\setminus\!(A\!+\!j)) - \frac{p}{c}\!\cdot\!|\ground\!\setminus\!A| + f(\ground\!\setminus\!A)
= -\frac{p}{c} + f_{\ground \setminus (A + j)}(j) 
\le -\frac{p}{c} - p \le 0.
\end{equation*}
Additionally, we note that $\hat{f}(\ground) = 0$, and so $\hat{f}$ must be nonnegative.  Finally, we show that $\hat{f}$ has curvature at most $c$.  We have:
\begin{align*}
\hat{f}_{\emptyset}(j) &= -\frac{p}{c} + f_{\ground - j}(j) \ge -\frac{p}{c} \\
\hat{f}_{\ground - j}(j) &= -\frac{p}{c} + f_{\emptyset}(j) \le -\frac{p}{c} + p = -\frac{1-c}{c}p,
\end{align*}
and therefore $\hat{f}_{\ground - j}(j)/\hat{f}_{\emptyset}(j) \le 1-c$.
\end{proof}
\end{document}